\documentclass[sigconf]{acmart}

\AtBeginDocument{%
  \providecommand\BibTeX{{%
    \normalfont B\kern-0.5em{\scshape i\kern-0.25em b}\kern-0.8em\TeX}}}

\setcopyright{iw3c2w3}
\copyrightyear{2020}
\acmYear{2020}
\acmConference[WWW '20]{Proceedings of The Web Conference 2020}{April 20--24, 2020}{Taipei, Taiwan}
\acmBooktitle{Proceedings of The Web Conference 2020 (WWW '20), April 20--24, 2020, Taipei, Taiwan}
\acmPrice{}
\acmDOI{10.1145/3366423.3380264}
\acmISBN{978-1-4503-7023-3/20/04}

\usepackage{booktabs} 

\usepackage{libertine}

\usepackage[ruled, noend, noline, linesnumbered]{algorithm2e}

\usepackage{bbm}
\usepackage{mathrsfs}
\usepackage{multicol}
\usepackage{multirow}
\usepackage{epstopdf}
\usepackage{color}
\usepackage{xcolor}
\usepackage{framed}
\usepackage{csvsimple}
\usepackage{longtable}
\usepackage[none]{hyphenat}

\usepackage[font={small}]{caption}
\usepackage{subcaption}

\newcommand{\mainalg}{Inverse-TS}
\newcommand{\sample}{Sample}
\newcommand{\modifiedts}{PEANUTS}
\newcommand{\ptsfinder}{PrefixedTur\'anShadowFinder}

\newcommand{\funckoc}{Func-$(k,1)$-Clique}
\newcommand{\funcktc}{Func-$(k,2)$-Clique-Type1}
\newcommand{\funcktcc}{Func-$(k,2)$-Clique-Type2}

\newlength{\dhatheight}

\newtheorem{theorem}{Theorem}[section]
\newtheorem{claim}[theorem]{Claim}

\newcommand{\ignore}[1]{}

\DeclareRobustCommand*\cal{\@fontswitch\relax\mathcal}

\newcommand{\cK}{\mathcal{K}}

\newcommand{\R}{\mathbb R}

\newcommand{\eps}{\varepsilon}

\newcommand{\bS}{\boldsymbol{S}}
\newcommand{\bT}{\boldsymbol{T}}

\newcommand{\NN}{\mathbb{N}}

\newcommand{\Exp}{\EX}

\newcommand{\EX}{\hbox{\bf E}}

\newcommand{\sz}[1]{\mathrm{size}(#1)}

\newcommand{\hc}{$h$-clique}
\newcommand{\hcs}{$h$-cliques}
\newcommand{\kc}{$k$-clique}
\newcommand{\kcs}{$k$-cliques}
\newcommand{\nkcs}{near-cliques}
\newcommand{\nkc}{near-clique}
\newcommand{\koc}{$(k,1)$-clique}
\newcommand{\kocs}{$(k,1)$-cliques}
\newcommand{\ktc}{$(k,2)$-clique}
\newcommand{\ktcs}{$(k,2)$-cliques}
\newcommand{\ts}{Tur\'anShadow}

\newcommand{\ps}{Prefixed-Shadow}

\newcommand{\pts}{Prefixed-Tur\'an-Shadow}

\newcommand{\Sec}[1]{\hyperref[sec:#1]{\S\ref*{sec:#1}}} 
\newcommand{\Eqn}[1]{\hyperref[eqn:#1]{(\ref*{eqn:#1})}} 
\newcommand{\Fig}[1]{\hyperref[fig:#1]{Fig.\,\ref*{fig:#1}}} 
\newcommand{\Tab}[1]{\hyperref[tab:#1]{Tab.\,\ref*{tab:#1}}} 
\newcommand{\Thm}[1]{\hyperref[thm:#1]{Theorem\,\ref*{thm:#1}}} 
\newcommand{\Fact}[1]{\hyperref[fact:#1]{Fact\,\ref*{fact:#1}}} 
\newcommand{\Lem}[1]{\hyperref[lem:#1]{Lemma\,\ref*{lem:#1}}} 
\newcommand{\Prop}[1]{\hyperref[prop:#1]{Prop.~\ref*{prop:#1}}} 
\newcommand{\Prob}[1]{\hyperref[prob:#1]{Prolem~\ref*{prob:#1}}} 
\newcommand{\Cor}[1]{\hyperref[cor:#1]{Corollary~\ref*{cor:#1}}} 
\newcommand{\Conj}[1]{\hyperref[conj:#1]{Conjecture~\ref*{conj:#1}}} 
\newcommand{\Def}[1]{\hyperref[def:#1]{Definition~\ref*{def:#1}}} 
\newcommand{\Alg}[1]{\hyperref[alg:#1]{Alg.~\ref*{alg:#1}}} 
\newcommand{\Ex}[1]{\hyperref[ex:#1]{Ex.~\ref*{ex:#1}}} 
\newcommand{\Clm}[1]{\hyperref[clm:#1]{Claim~\ref*{clm:#1}}} 
\newcommand{\Obs}[1]{\hyperref[obs:#1]{Obs.~\ref*{obs:#1}}} 
\newcommand{\Step}[1]{\hyperref[step:#1]{Step~\ref*{step:#1}}} 

\begin{document}

\title{Provably and Efficiently Approximating Near-cliques using the Tur\'an Shadow: PEANUTS} 

\author {Shweta Jain} 
\affiliation{University of California, Santa Cruz}
\affiliation{Santa Cruz, CA, USA}
\email{sjain12@ucsc.edu}

\author {C. Seshadhri} 
\affiliation{University of California, Santa Cruz}
\affiliation{Santa Cruz, CA}
\email{sesh@ucsc.edu}

\begin{abstract} Clique and near-clique counts are important graph properties with applications in graph generation, graph modeling, graph analytics, community detection among others. They are the archetypal examples of dense subgraphs. While there are several different definitions of near-cliques, most of them share the attribute that they are cliques that are missing a small number of edges. Clique counting is itself considered a challenging
problem. Counting near-cliques is significantly harder more so since the search space for near-cliques is orders of magnitude larger than that of cliques.  

We give a formulation of a near-clique as a clique that is missing a constant number of edges. We exploit the fact that a near-clique contains a smaller clique, and use techniques
for clique sampling to count near-cliques. This method allows us 
to count near-cliques with 1 or 2 missing edges, in graphs with tens of millions of edges. 
To the best of our knowledge,
there was no known efficient method for this problem, and we obtain a $10x-100x$ speedup over existing algorithms for counting near-cliques.

Our main technique is a space efficient adaptation of the Tur\'{a}n Shadow sampling approach,
recently introduced by Jain and Seshadhri (WWW 2017). This approach constructs a large
recursion tree (called the Tur\'{a}n Shadow) that represents cliques in a graph. We design a novel algorithm that builds an estimator for near-cliques, using a online, compact construction of the Tur\'{a}n Shadow.



\end{abstract}



\maketitle
\section{Introduction} \label{sec:intro}

Subgraph counting is an important tool in graph analysis that goes by many names such as motif
counting, graphlet analysis, and pattern counting. The aim is to count the number of occurrences of a (generally small) subgraph in a much larger graph. Among these subgraphs,
cliques are arguably the most important. Even the simplest clique, the triangle, has
a rich history of algorithms and applications. There has been much recent
focus on counting cliques in large graphs~\cite{DBS18,JaSe17,FFF15, HaSe16,WLRTZG14,BRRA12}.



While clique counts are important, the requirement that every edge in the clique be present is excessively rigid. Data is often noisy or incomplete, and it is likely that
cliques that are missing even an edge or two are significant.
Hence, it is important to also look at counts of patterns that are extremely close to being cliques. 
We will call these structures near-cliques but they are also known as quasi-cliques~\cite{LW08,PYB12} and defective cliques~\cite{YPTG06} and have several applications ranging from clustering to prediction. Recent work on 
has used the fraction of \nkcs{} to \kcs{} to define
higher order variants of clustering coefficients~\cite{YBLG17}.

In the bioinformatics literature, near-cliques (or defective cliques, as they are known) have been used to predict missed protein-protein interactions in noisy PPI networks~\cite{YPTG06} and have been shown to have good predictive performance. An alternative viewpoint of looking at near-cliques views them as dense subgraphs. Mining dense subgraphs is an important problem with many applications in Network Analysis. ~\cite{fratkin06,sariyuce15,kumar99,chen10,alvarez06}

Counting cliques is already challenging, and counting near-cliques introduces more
challenges. Most importantly, near-cliques do not enjoy the recursive structural property of cliques - that a subset of a clique is also a clique. This rules out most recursive
backtracking algorithms for clique counting. Moreover, empirical evidence suggests that the number of \nkcs{} in real world datasets is order of magnitudes higher than that of cliques, making the task of counting them equally difficult if not more. \Fig{ncs-vs-cliques} shows the ratio of 3 different types of \nkcs{} to the number of \kcs{} for $k=5$ for 4 real world graphs. The number of \nkcs{} is often ten times higher than the number of \kcs{}. 

There are several different ways of defining near-cliques. ~\cite{Tsourakakis13} define $\alpha$-quasi-cliques as cliques that are missing a $\alpha$ fraction of the edges. Other formulations define them in terms of graph properties like degree of every vertex in the near-clique or diameter of the near-clique. A set $S$ of size $n$ is called a $k-$plex if every member of the set is connected to $n-k$ others. A $k$-club is a subset $S$ of nodes such that in the subgraph induced by $S$, the diameter is $k$ or less.  All these formulations have the common property that they represent a clique that is missing a few edges. We formulate near-cliques in a slightly different way, as cliques that are missing 1 or 2 edges. The advantage of defining them this way is that they allow us to leverage the machinery of clique counting. Every such near-clique has a smaller clique contained in it. By sampling the smaller cliques and using them as hints to find near-cliques, we give an estimate for the total number of near-cliques. In \Sec{practice} we show an interesting application of such near-cliques where we run our algorithm on a citation network to discover papers that perhaps should have cited other papers but did not.

\begin{figure*}[t!]   
       \begin{subfigure}[b]{0.24\textwidth}
        \centering
        \includegraphics[width=\textwidth]{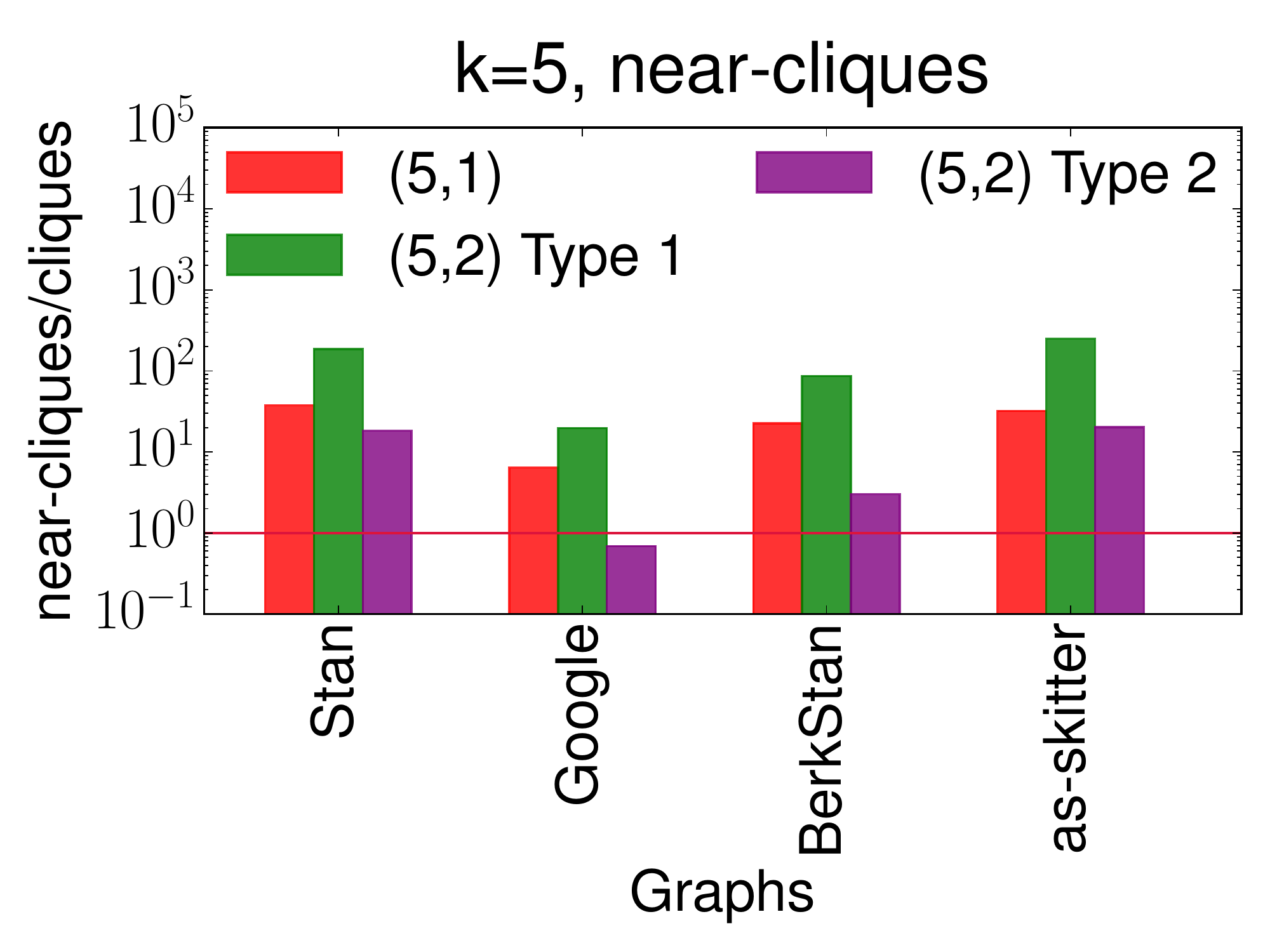}
        \caption{Ratios}
        \label{fig:ncs-vs-cliques}
        \end{subfigure}
        ~
        \begin{subfigure}[b]{0.24\textwidth}
        \centering
        \includegraphics[width=\textwidth]{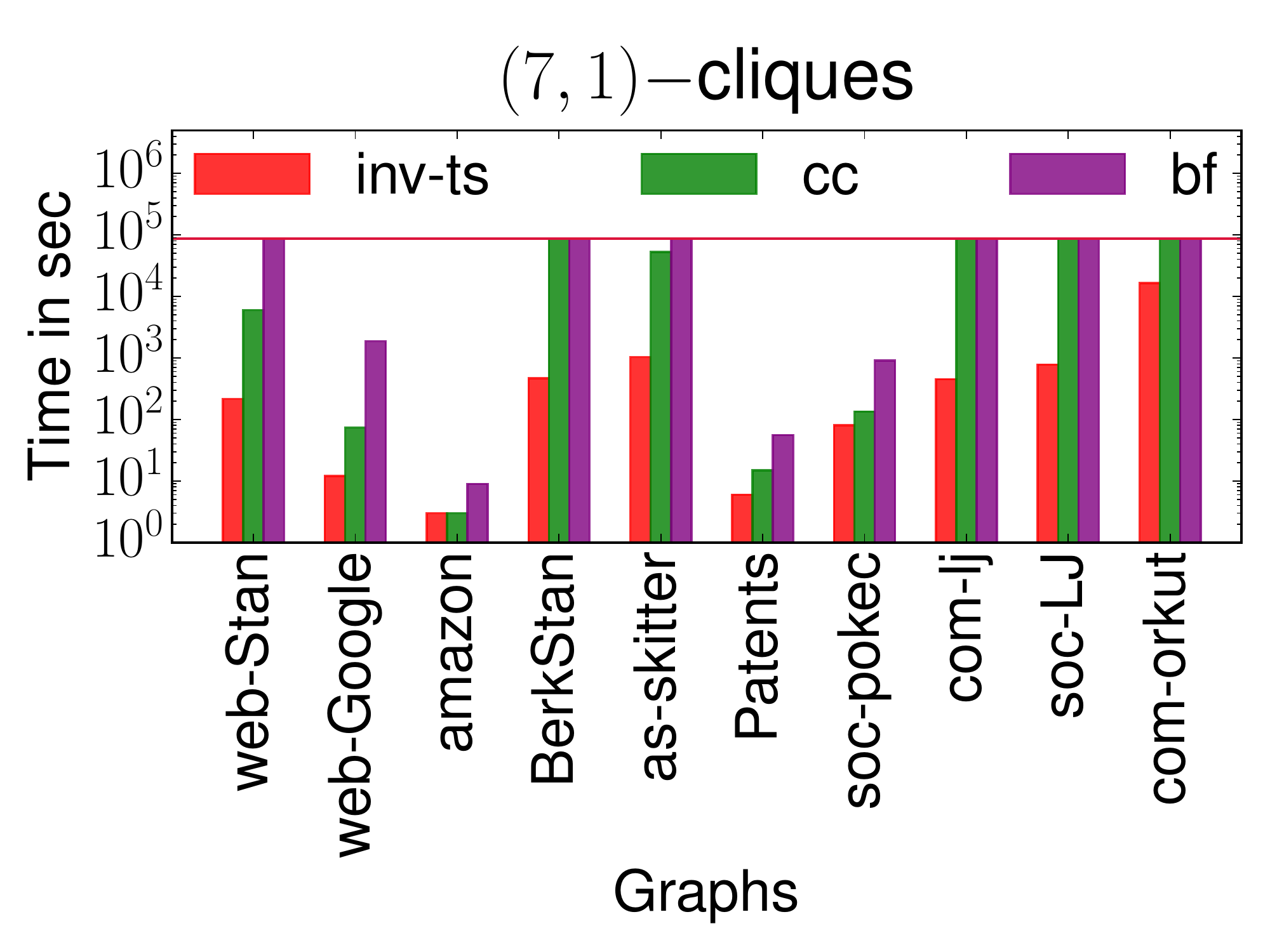}
        \caption{Timings}
        \label{fig:timing-7-nc}
        \end{subfigure}
        ~   
        \begin{subfigure}[b]{0.24\textwidth}
        \centering
        \includegraphics[width=\textwidth]{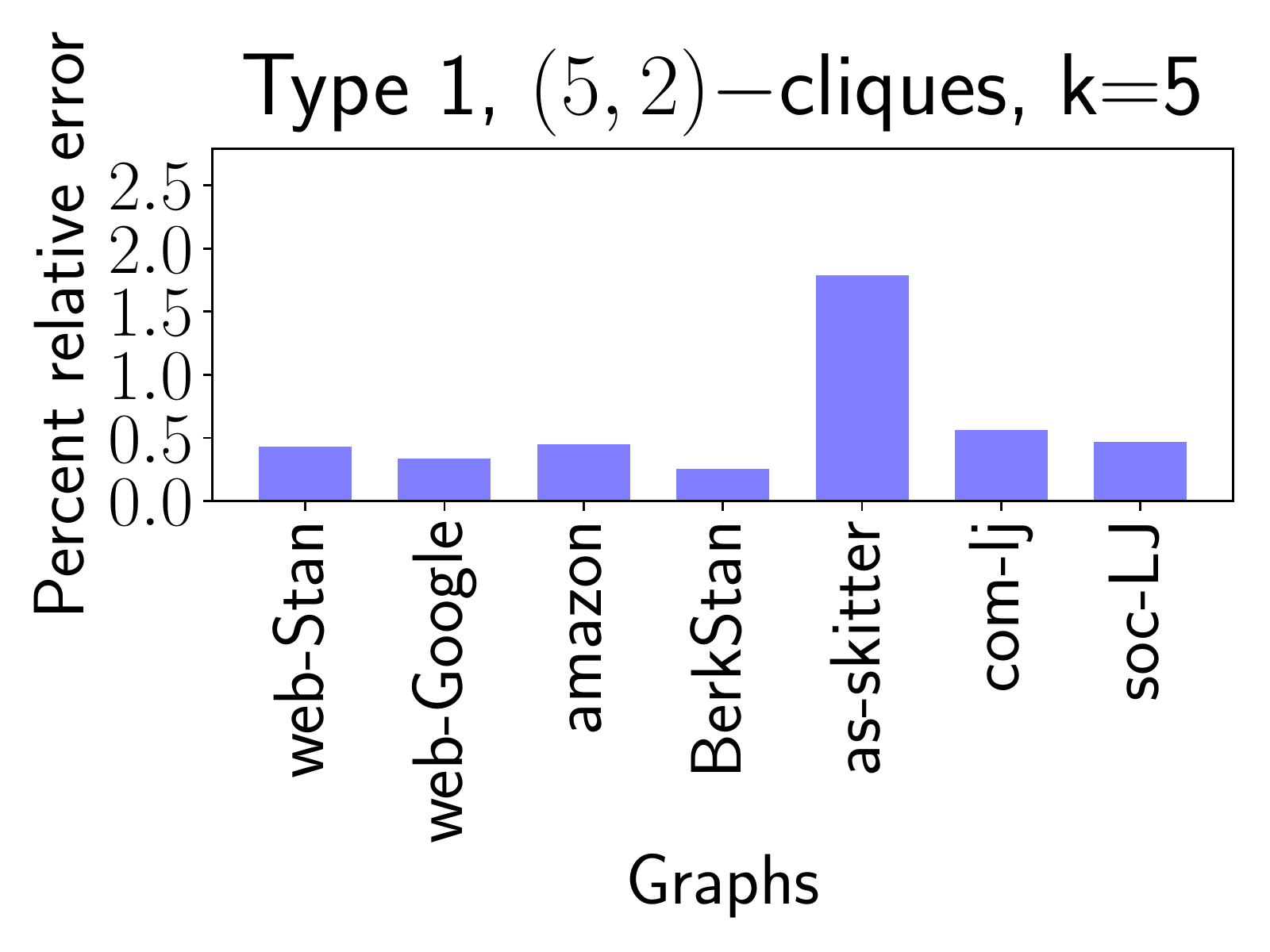}
        \caption{Error}
        \label{fig:rel-error-5-nc2}
        \end{subfigure}
        ~
        \begin{subfigure}[b]{0.24\textwidth} 
        \centering
        \includegraphics[width=\textwidth]{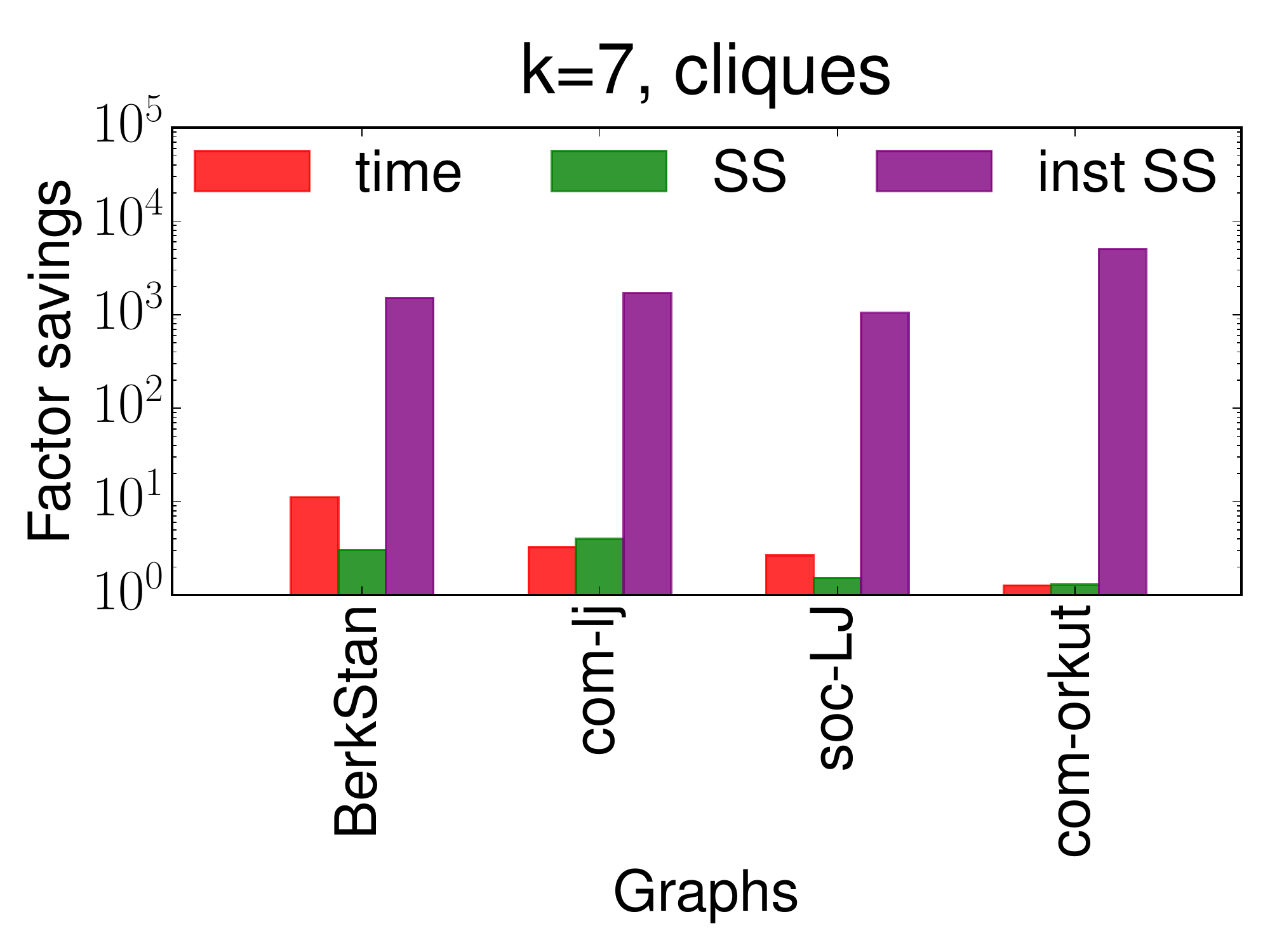}
        \caption{TS vs \mainalg{}}
        \label{fig:ts-vs-inv-ts}
        \end{subfigure} 

           \caption{\Fig{ncs-vs-cliques} shows the ratio of number of different types of \nkcs{} to \kcs{} for $k=5$ in four real world graphs. The red line indicates ratio = $1$. In most cases the number of \nkcs{} is at least of the same order of magnitude as number of \kcs{}, if not more. \Fig{timing-7-nc} shows the time required by \mainalg{} (inv-ts), color-coding (cc) and brute force (bf) to estimate the number of $(7,1)$-cliques in 10 real world graphs. The $y-$axis shows time in seconds on a log scale. The red line indicates 86400 seconds (24 hours). All experiments that ran for more than 24 hours were terminated. \mainalg{} terminated in minutes in all cases except com-orkut, giving a speedup of anywhere between 3x-100x. \Fig{rel-error-5-nc2} shows the percentage error in the estimates for Type 1 \ktcs{} for $k=5$ obtained using \mainalg{}. As we can see, the error is $<2\%$ and in most cases $<1\%$. \Fig{ts-vs-inv-ts} shows the savings in time and space when using \mainalg{} (500000 samples) vs when using \ts{} (50000 samples) to estimate the number of 7-cliques in 4 of the largest real world graphs we experimented with. The green bars show the factor savings in the percentage of the Tur\'an Shadow that was explored (factor of 2-10). The purple bar shows the factor saving in the maximum amount of space required for the Tur\'an Shadow at any instant.}
        

\end{figure*}

\begin{figure}
 \begin{subfigure}[b]{0.155\textwidth}
    \centering
    \includegraphics[width=\textwidth]{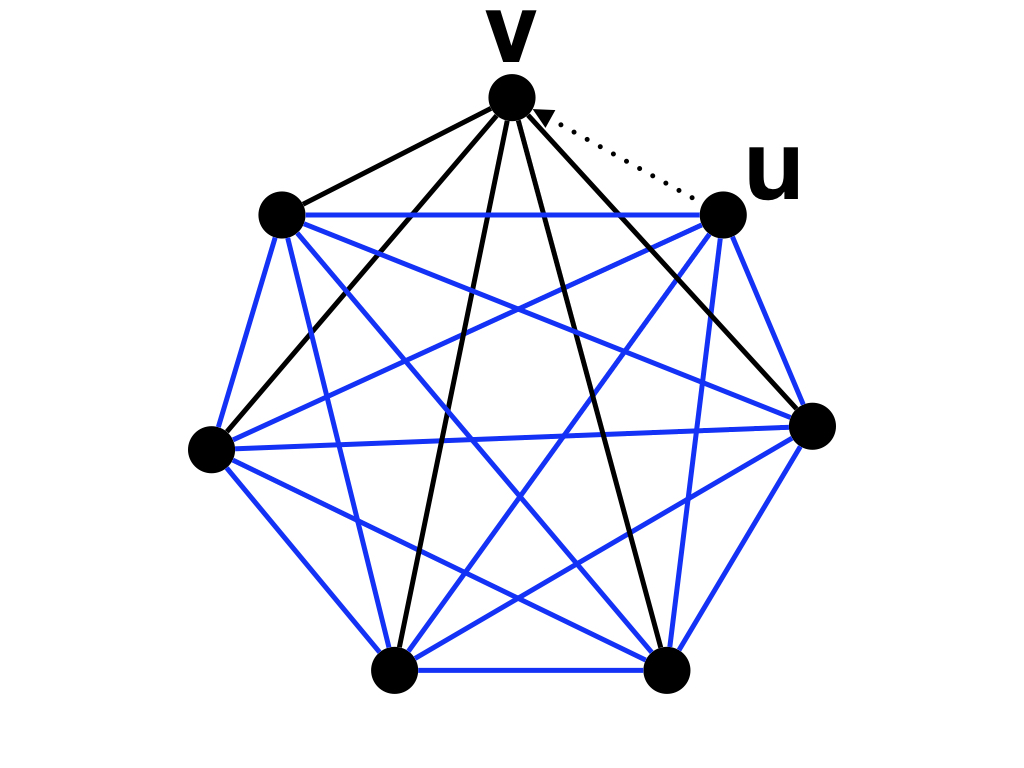}
    \caption{$(7,1)$}
    \label{fig:koc-annotated}
    \end{subfigure}
    \begin{subfigure}[b]{0.155\textwidth}
    \centering
    \includegraphics[width=\textwidth]{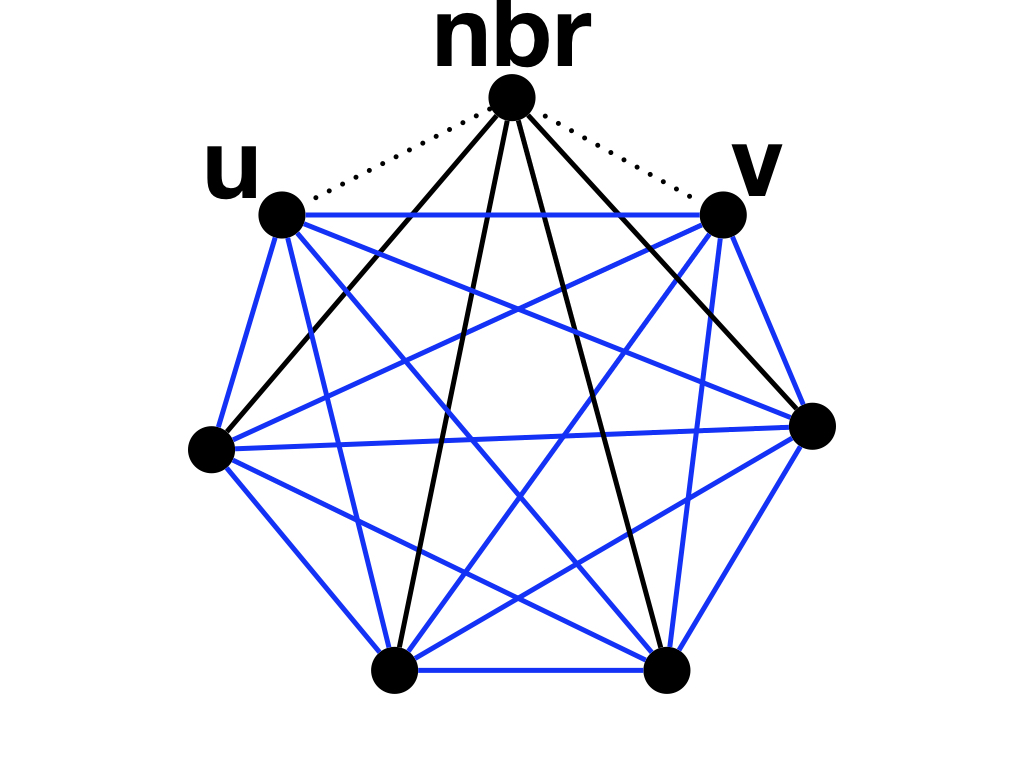}
    \caption{Type 1 $(7,2)$}
    \label{fig:type1-annotated}
    \end{subfigure}
    \begin{subfigure}[b]{0.156\textwidth}
    \centering
    \includegraphics[width=\textwidth]{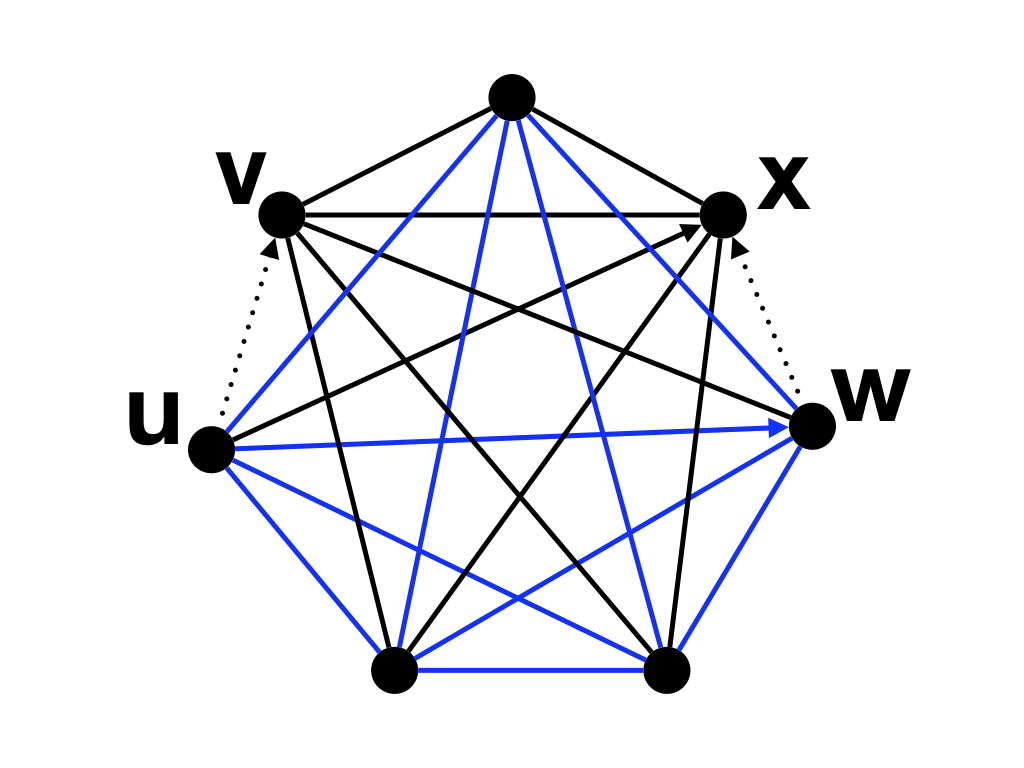}
    \caption{Type 2 $(7,2)$}
    \label{fig:type2-annotated}
    \end{subfigure}
\caption{Near-$7$-cliques. Dotted lines indicate the missing edges. Blue lines mark the contained clique.}
\label{fig:types-annotated}        
\end{figure}

\subsection{Problem description} \label{sec:problem}

A \kc{} is a set of $k$ vertices such that there is an edge between all pairs of vertices belonging to the set. We define \koc{} and \ktc{} below. For the rest of this paper, whenever we say near-cliques, we will imply the following 3 kinds of near-cliques (unless mentioned otherwise)

\begin{definition}
A \koc{} is a \kc{} with exactly 1 edge missing. 
\end{definition}

For \ktcs, there are 2 configurations possible - one in which the missing edges share a vertex, and one in which they don't.

\begin{definition}
A Type 1 \ktc{} is a \kc{} with exactly 2 edges missing such that the missing edges share a vertex. 
\end{definition}

\begin{definition}
A Type 2 \ktc{} is a \kc{} with exactly 2 edges missing such that the missing edges do not share a vertex. 
\end{definition}



The different types of \nkcs{} are shown in \Fig{types-annotated}. We want to estimate the number of \kocs{} and \ktcs{} in $G$. Note that all our near-cliques are induced and obtaining counts of non-induced near-cliques is simply a matter of taking a linear combination of the number of \kcs{} and \nkcs{}. For the sake of brevity, we skip a detailed discussion.



We
stress that we make no distributional assumption on the
graph. All probabilities are over the internal randomness of
the algorithm itself (which is independent of the instance).

\subsection{Our contributions} \label{sec:contribute}
We provide a randomized algorithm based on \ts{} called \modifiedts{}  which estimates the counts of \kocs{} and \ktcs{}. In addition, we also provide a heuristic algorithm called \mainalg{} based on \modifiedts{} which takes roughly the same time as \modifiedts{} (and in some cases, upto 10x less time) but drastically reduces the space required. Our implementation of \mainalg{} on a commodity machine showed significant savings in terms of time in obtaining counts of \nkcs{} over other methods like color-coding and brute force counting and showed consistently low error over 100s of runs of the algorithm.
    
    \textbf{Leveraging cliques for near-cliques:} 
    Data being noisy, cliques are brittle and as a result, number of \nkcs{} is often very large. However, it is not at all clear how one can count their number without looking at every set of $k-$vertices, which is computationally very expensive. \modifiedts{} uses the fact that near-cliques themselves contain cliques, and leverages \ts{} to count near-cliques. There exist algorithms for generic pattern counting which can be used for counting near-cliques but there is no known algorithm dedicated to finding near-cliques that exlploits the clique-like structure of near-cliques to give a faster estimate. 

    \textbf{Extremely fast:} 
    \modifiedts{} is based on the observation that every near-clique contains a smaller clique. Thus, we can use cliques as clues for finding near-cliques. We leverage a fast clique-counting algorithm (\ts{}) to achieve fast and accurate near-clique counting. 
    \Fig{timing-7-nc} shows the time taken by \mainalg{}, color-coding (cc) and brute force (bf) to count the number of $(7,1)$-cliques for a variety of graphs. \mainalg{} is able to estimate their number to within 2\% error in a graph (com-lj) with 4 million vertices and 34 million edges in 452 seconds which is at least 100 times faster than cc and bf. As we will show later, similar performance is found in the estimation of other near-cliques and on other graphs. 

    \textbf{Extremely accurate:} Similar to \ts{}, \mainalg{} uses the seminal result from extremal combinatorics, called Tur\'an's theorem which allows for efficiently sampling cliques, which translates to fast and accurate estimation of the number of near-cliques. \Fig{rel-error-5-nc2} shows the error in the estimate obtained for number of Type 2 $(5, 2)-$cliques (a specific configuration of $(5, 2)-$cliques) in a variety of graphs using \mainalg{}. As we can see, all the errors were within 2\%. Moreover, unlike color-coding, \mainalg{} allows us to control the number of samples we take, and even using 500K samples, \mainalg{} was more accurate and took less time than color-coding (\ref{sec:results}). 

    For many of the graphs we experimented with, the brute force algorithm had not terminated within 1 day and thus was unable to give us ground truth values, but in the cases where the algorithm did terminate, we saw that \mainalg{} gave $<5\%$ error and mostly $<2\%$. For the cases where the brute force algorithm did not terminate, we looked at the output of 100 runs of our algorithm. In all cases, the algorithm showed very good convergence properties (more details in \Sec{results}).


    \textbf{Excellent space efficiency:} 
    \ts{}  requires that the entire shadow be generated and stored, which for a graph with 100s of millions of edges can potentially require large amount of memory. Our practical implementation of \mainalg{} addresses this by removing the separation in the Shadow construction and sampling phases and instead, performs sampling while the shadow is being constructed in an online fashion. This eliminates the need for storing the entire Shadow and consequently gives savings of orders of magnitude in space required. The purple bars in \Fig{ts-vs-inv-ts} show the factor savings in the maximum shadow size required to be stored at any point (instantaneous shadow size or inst SS) for \mainalg{} vs the space required by \ts{}. There is atleast 100x savings in space using \mainalg{}.
\textbf{Comparison with other algorithms:} 
We do a thorough analysis of \mainalg{} by deploying it on a number of real-world graphs of varying sizes. In most cases we observed that \mainalg{} was considerably fast while showing consistently low error over 100s of runs of the algorithm. We also do a thorough comparison of \mainalg{} with other generic pattern-counting algorithms like color-coding. \Fig{timing-7-nc} shows the time required for counting $(7,1)$-cliques by the different methods. Across all of our experiments we observe that \mainalg{} was at least 10 times faster on most graphs as compared to other algorithms. 

\textbf{All code and data available:} All the datasets we used are publicly available at ~\cite{Snap}. In addition, we can readily make the code for our algorithm publicly available if the paper is accepted for publication.


\subsection{Related Work}

Pattern counting, also known as graphlet counting or motif counting has been an important tool for graph analysis. It has been used in bioinformatics~\cite{Milo2002,Wernicke06,Pr07}, social sciences~\cite{HoLe70}, spam detection~\cite{BeBoCaGi08}, graph modeling~\cite{SeKoPi11}, etc. Triangle counting, and more recently, clique counting have gained a lot of attention~\cite{DBS18,JaSe17,FFF15} due to their special role in characterizing real-world graphs. Clique counts have been employed in applications such as discovery of dense subgraphs~\cite{SaSePi14,Ts15}, in topological approaches to network analysis~\cite{SGB16}, graph clustering~\cite{YBLG17} among others. More generally, motif counts have been used in clustering~\cite{YBLG17,TsPaMi16}, evaluation of graph models~\cite{SPR17,SeKoPi11}, classification of graphs~\cite{UganderBK13} etc. 

On the theoretical side, several motif-counting algorithms exist~\cite{ChNi85,YBLG17,CDM17}. On the more practical side, only recently, efficient methods for counting graphlets upto size 5~\cite{WZ+18,PiSeVi17,JhSePi15,HoDe17} have been proposed. Most of these are extensions of triangle counting methods and do not scale. For patterns of larger sizes, two widely used techniques are the MCMC~\cite{HaSe16,WLRTZG14} and color-coding (CC) of~\cite{AlYuZw94}. However, as shown in ~\cite{BCKLP18}, MCMC based methods have poorer accuracy for the same running time than CC and for patterns of sizes greater than 5, CC is also generally quite inefficient, as we will show in our results. Motif counting has been studied in streaming ~\cite{BoDo+08,KMSS12} and distributed settings ~\cite{EEKBD16} and in temporal networks~\cite{PBL17}.

All these methods are geared towards counting arbitrary patterns with upto 6 nodes but none of these methods scale beyond 6 nodes. Moreover, these are generic pattern counting methods that do not utilize the clique-like nature of near-cliques to give more efficient methods. Ours is the first work to do so.



\textbf{Dense subgraph algorithms:} The notion of dense subgraphs as near-cliques was introduced by Tsourakakis et. al. in ~\cite{Tsourakakis13}. There are several different formulations of dense subgraphs, many of which are NP-Hard (indeed, even the problem of finding the densest subgraph on $k$ vertices, known as the densest-$k$-subgraph is NP-Hard ~\cite{SaSePi14}). The algorithms of Andersen and Chellapilla ~\cite{AnCh09}, Rossi et al. ~\cite{RossiG15}, and Tsourakakis et al. ~\cite{Tsourakakis13,Ts15} provide practical algorithms for some of the formulations. However, most of them focus on finding or approximating the densest subgraph rather than giving global stats. 

\section{Main ideas} \label{sec:main}

The starting point of our result is the \ts{} algorithm for estimating the number of \kcs{} in a graph. \ts{} is based on a seminal theorem of Tur\'an and Erd\"os that says that: if the edge density of an $n-$vertex graph is greater than $1 - 1/(k-1)$ (the Turan density), then the graph is guaranteed to have many ($O(n^{k-2})$) $k$-cliques. This implies that if we randomly sample a $k$-vertex set from the graph, the probability of it being a $k$-clique would be high. \ts{} exploits this fact by splitting $G$ into (possibly overlapping) Turan-dense subgraphs such that there is a one-to-one correspondence between the cliques of a specific size in each subgraph, and the number of $k$-cliques in G. The set of all such subgraphs of $G$ is called the Tur\'an Shadow of $G$. Essentially, \ts{} reduces the search space for $k$-cliques in $G$ from 1 large sparse graph to several dense subgraphs. 

More importantly though, for any $h$, \ts{} provides an efficient way of sampling a u.a.r. \hc{} from $G$. Let $C_h$ be the set of all \hcs{} in $G$ and let $f:C_h \rightarrow \R^+$ be a bounded function over all \hcs{}, then we can obtain an unbiased estimate for $F = \sum\limits_{K \in C_h}{f(K)}$ by obtaining the average of $f$ over a set of uniformly sampled \hcs{} and scaling by the total number of \hcs{}. In other words, we can use this clique sampler to obtain an unbiased estimate of the sum (and mean value) of any bounded function over \hcs{}. We exploit this fact to obtain an estimate of the number of \nkcs{}. 

To estimate the number of \kocs{}, we make the following observation: Every \koc{} has exactly two $k-1$-cliques embedded in it. Let $C_{k_{-1}}$ be the set of \kocs{} in $G$ and $C_{k-1}$ be the set of $k-1$-cliques in $G$, and $ \forall K \in C_{k-1}$, let $f(K) = $ number of \kocs{} that clique $K$ is contained in, then $\sum\limits_{K \in C_{k-1}}{f(K)} = 2|C_{k_{-1}}|$. However, since every \koc{} is counted twice, the variance of the estimator can be pretty large. We observe that if the missing edge in a \koc{} is $(u,v), u<v$, exactly one of the $k-1$-cliques contains $u$ and the other contains $v$. In order to reduce the variance, we define $f(K) = $ number of \kocs{} that clique $K$ is contained in, such that $u \in K$ i.e. we break ties based on the direction of the missing edge. With this formulation, $\sum\limits_{K \in C_{k-1}}{f(K)} = |C_{k_{-1}}|$. 

For \ktcs{}, there are 2 possible configurations, as shown in \Fig{types-annotated}. Type 1 consists of exactly one $k-1$-clique embedded in it. Hence, we set $f(K)=$ number of Type 1 \ktcs{} that a given $k-1$-clique $K$ is contained in. Type 2 \ktcs{} are a bit more complicated. A Type 2 \ktc{} has exactly four $k-2$-cliques embedded in it. If the edges $(u,v)$ and $(w,x)$ are missing, then there is an induced cycle involving $u,v,w$ and $x$ and every edge of this cycle gives a different $k-2$-clique of the four $k-2$-cliques embedded in the \ktc{}. Let $min(u,v,w,x)=u$ and let $min(w,x)=w$. Then, for $k-2$-clique $K$, we set $f(K)=$ number of \ktcs{} such that $u, w \in K$. 

As long as $f$ is bounded and is a ``well behaved function'' i.e. has low variance, we can efficiently estimate $F$ using \ts{} as a black box. Improving the running time of the black box only improves the running time of the overall algorithm. We observe that in \ts{}, most of the time is spent in constructing the Shadow, but only a small fraction of it is used to gather samples. Thus, if we can first sample and determine which areas of the Shadow the samples lie in, we can save time by developing only those parts of the Shadow instead of developing the whole Shadow. Additionally, when the number of samples are fixed (as is the case in the practical implementation of our algorithm), we can interleave the development of the parts of the Shadow with sampling for \hcs{} from those parts, thus obtaining our estimate of $F$ in an online fashion. This leads to considerable savings in space and time.

Outline: In \Sec{prelim} we set some basic notation. In \Sec{inverse-ts} we show our basic framework \modifiedts{} and an optimized version of it called \mainalg{}. Depending on which type of pattern we want to count, we propose and analyze different counters in \Sec{counters}. Finally, in \Sec{results} we provide a detailed experimental study of \mainalg{} and its comparison with the state-of-the-art.

\section{Preliminaries} \label{sec:prelim}

We set some notation. The input graph $G$ has $n$ vertices and $m$ edges.
We will assume that $m \geq n$. Let $\alpha$ be the degeneracy of the graph. Recall that the degeneracy is the maximum outdegree of any vertex when the edges of the graph are oriented according to the degeneracy ordering of the vertices in $G$. Let  $N_{v}(G)$ represent the neighborhood of $v$ and let $N_{v}^+(G)$ represent the outneighborhood of $v$ when the vertices are ordered by degeneracy.

We use ``u.a.r."
as a shorthand for ``uniform at random".

We will be using the following (rescaled) Chernoff bound.

\begin{theorem} \label{thm:chernoff} [Theorem 1 in~\cite{DuPa09}] Let $X_1, X_2, \ldots, X_k$
be a sequence of iid random variables with expectation $\mu$. Furthermore, let $X_i \in [0,B]$. Then, for $\eps< 1$, $ \Pr[|\sum_{i=1}^k X_i - \mu k| \geq \eps\mu k] \leq 2\exp(-\eps^2 \mu k/3B) $.
\end{theorem}
\section{Main algorithm} \label{sec:inverse-ts}

 At the core of \ts{} lies an object called the shadow. We define an analogous structure called \ps{}. 


\begin{definition}
Let $C_k(G)$ be the set of all $k-$cliques in $G$.  A $k$-clique \ps{} $\bS$ for graph $G$ 
is a set of triples $\{(P_i, S_i, \ell_i)\}$ where $P_i \subseteq V$, $S_i \subseteq V$ and $\ell_i \in \NN$ such that $ \forall (P_i, S_i, \ell_i) \in \bS, \forall c \in C_{\ell_i}(S_i), P_i \cup c$ is a unique \kc{} in $G$ and there is a bijection between $C_k(G)$ and $ \bigcup \limits_{(P_i, S_i, \ell_i) \in \bS}\bigcup \limits_{c \in C_{\ell_i}(S_i)} P_i \cup c$.

Moreover, if the multiset $\{(S_i,\ell_i)\}$ is such that $ \forall (S_i, \ell_i), \rho_2(S_i) > 1 - 1/(\ell_i - 1)$ where $\rho_2(S_i)$ represents the edge density of $S_i$, then $\bS$ is a \kc{} \pts{} of $G$.
\end{definition}

It is easy to see that $\{(v, N^+_v, h-1)\}$ is an $h$-clique \ps{} of $G$. 



We will briefly recap how \ts{} constructs the shadow. It orders the vertices of $G$ by degeneracy and converts it into a DAG. As shown in ~\cite{FFF15}, to count \kcs{} in $G$ it suffices to count the number of $k-1$-cliques in the outneighborhood of every vertex. Hence,  for every vertex $v \in V$, \ts{} counts the number of $k$-cliques with $v$ as the lowest order vertex by looking at the number of $k-1$-cliques in the outneighborhood of $v$, and it applies this procedure recursively. When the outneighborhood becomes dense enough, instead of continuing to expand the partial clique, it adds the outneighborhood to the shadow and continues until there are no more outneighborhoods left to be added to the shadow.

Algorithm \ptsfinder{} carries out exactly the same steps as \textit{Shadow-Finder} in ~\cite{JaSe17}, except that at each stage it also maintains the partial clique $P$. 

\begin{algorithm}
\caption{\ptsfinder$(G,k)$}
 Initialize $\bT = \{(\emptyset, V,k)\}$ and $\bS = \emptyset$ \\
 While $\exists (P,S,\ell) \in \bT$ such that $\rho_2(S) \leq 1 - \frac{1}{\ell-1}$ \\
 \ \ \ \ Construct the degeneracy DAG $D(G|_S)$ \\
 \ \ \ \ Let $N^+_s$ denote the outneighborhood (within $D(G|_S)$) of $s \in S$ \\
 \ \ \ \ Delete $(P,S,\ell)$ from $\bT$ \\
 \ \ \ \ For each $s \in S$ \\
 \ \ \ \ \ \ \ If $\ell \leq 2$ or $\rho_2(N^+_s) > 1 - \frac{1}{\ell-2}$ \\
 \ \ \ \ \ \ \ \ \ \label{step:add} Add $(P \cup \{s\}, N^+_s,\ell-1)$ to $\bS$ \\ 
 \ \ \ \ \ \ \ \label{step:move} Else, add $(P \cup \{s\}, N^+_s,\ell-1)$ to $\bT$ \\
 Output $\bS$
 \end{algorithm}

 \begin{claim} \label{clm:ptsf}
 Given a graph $G$ and integer $k$, \ptsfinder{} returns a $k$-clique \pts{} of $G$. Its running time is $O(|V|^{k+1})$.
 \end{claim}
\begin{proof}
 When the function returns, $T$ is empty, and any element $(P,S,\ell) \in \bS$ was added to $\bS$ only when $\rho_2(S) > 1-1/(\ell-1)$. Thus, if $\bS$ is a \ps{}, it is also a \pts{}.

By Theorem 5.2 in ~\cite{JaSe17}, multiset $\{(S, \ell)\}$ is a shadow and hence, there is a bijection between $C_k(G)$ and $\bigcup_{(P,S,\ell) \in \bS} C_\ell(S)$. Thus, it suffices to prove that $ \forall (P, S, \ell) \in \bT \cup \bS, \forall c \in C_{\ell}(S), P \cup c$ is a unique \kc{} in $G$. We will prove this using induction. At the start of the first iteration, $P$ is empty, $S=V$ and $\ell = k$, $\bS = \{(P,S,\ell)\}$ and $\bT $ is empty. Thus, for the base case, the hypothesis is trivially true.

Suppose the hypothesis is true at the start of some iteration and lets say element $E=(P',S',\ell')$ is deleted from $\bT$ at the start of this iteration. Each $E_s=(P' \cup \{s\}, N^+_s,\ell'-1)$ for $s \in S'$ is added to $\bS$ or to $\bT$. Let $\cK(E)=\{P'\cup c | c \in C_{\ell'}(S')\}$ denote the set of \kcs{} obtained from $E$. It suffices to prove that: (i) for any $k-$clique $K \in \cK(E), K \in \bigcup_s \cK(E_s)$, (ii) $|\cK| = \sum_s |\cK(E_s)|$. 

Consider a \kc{} $K = P' \cup c, c \in C_{\ell'}(S')$. Let $s$ be the lowest order vertex in $c$ according to the degeneracy ordering in $G|_{S'}$. Then, $c \setminus \{s\}$ is an $\ell-1$-clique in $N_{s}^+$. Thus, $K \in  \cK(E_{s})$. Additionally, for $c \in C_{\ell'}(S')$ the smallest vertex in $c$ defines a partition over $C_{\ell'}(S')$. Hence, $|C_{\ell'}(S')|=\sum_{s \in S'}{|C_{\ell'-1}(N_s^+)|}$ i.e. $|\cK(E)| = \sum_{s \in S'} {|\cK(E_{s})|}$. Hence, proved. 

The out-degree of every vertex is at most $|V|$ and the depth of the recursive calls is atmost $k-1$. When processing an element $(P, S, \ell)$ it constructs the graph $G|_S$ which takes time atmost $|V|^2$ since it queries every pair of vertices in $S$ and $|S|<|V|$. Thus, the time required is $O(|V|^{k+1})$. 
 \end{proof}

 \begin{algorithm} 
\caption{\sample$(\bf{S})$ \newline 
{\bf Inputs:} \textit{$\bf{S}$}: $k-$clique \pts{} of some graph $G$ \newline
{\bf Output:} $B$: $k-$vertex set
} \label{alg:sample}
Let $w(\bS)=\sum\limits_{(P',S',\ell') \in \bf{S}}{{|S'| \choose \ell'}}$ \\
Set probability distribution $D$ over $\bf{S}$ such that $(P,S,\ell) \in \bf{S}$ is sampled with probability ${|S| \choose \ell}/w(\bS)$ \\
Sample a $(P,S,\ell)$ from $D$ \\
Choose a u.a.r. $\ell-$tuple $c$ from $S$ \\
Let $B = P \cup \{c\}$ \\
return $B$
\end{algorithm}

\begin{claim} \label{clm:uar}
 The probability of any $k$-clique $K$ in $G$ being returned by a call to \sample{} is $\frac{1}{w(\bS)}$ . 
 \end{claim}
\begin{proof}
Let $E = (P,S,\ell) \in \bS$ where $\bS$ is the $k$-clique \ps{} of some graph $G$. Note that $w(\bS)=\sum\limits_{(P',S',\ell') \in \bf{S}}{{|S'| \choose \ell'}}$.  Let $c$ be an $\ell-$clique in $S$ and let $K = P \cup c$ then $K$ must be a unique $k-$clique in $G$. 

$Pr(\text{K is sampled}) =$  $Pr(\text{E is sampled from D})* Pr(\text{c is sampled from S}) = \frac{{|S| \choose \ell}}{w(\bS)} * \frac{1}{{|S| \choose \ell}} = \frac{1}{w(\bS)}$.
Thus, every $k-$clique in $G$ has the same probability of being returned by \sample{}.  
 \end{proof}

 We will first describe \modifiedts{}. Essentially, it constructs the \pts{} of $G$, samples \hcs{}, obtains $f$ for the sampled \hc{} and estimates the value of $F$.

\begin{algorithm} 
\caption{\modifiedts$(G, h, s, Func)$ \newline 
{\bf Inputs:} \textit{$G$}: input graph, \textit{$h$}: clique size // $= k$ for cliques, $k-1$ for \koc{} and Type $1$ \ktc{}, $k-2$ for Type $2$ \ktc{} \newline
\textit{$s$}: budget for samples, \textit{$Func$}: Function that returns $f(K)$ for $h$-clique $K$. \newline
{\bf Output:} $\hat{F}$: estimated $F$
} \label{alg:modifiedts}
$\bf{S}= \ptsfinder(G, h)$\\ 
Let $w(\bS) = \sum\limits_{(P,S,\ell) \in \bf{S}}{|S| \choose \ell}$ \\
For $i = 1,2,...,s$: \\
\ \ \ \ $K=\sample(\bf{S})$ \\
\ \ \ \ If $K$ is a clique, set $X_i = Func(G, K)$ \\ 
\ \ \ \ else set $X_i = 0$ \\
\ \ \ \ $W = W + X_i$ \\
let $\hat{F} = \frac{W}{s}w(\bS)$ \\
return $\hat{F} $
\end{algorithm}

\begin{theorem} \label{thm:basic}
Let $f$ be a function over $h$-cliques, bounded above by $B$ such that given an $h$-clique, it takes $O(T_f)$ time to obtain the value of $f$. Let $\hat{F}$ be the output of \modifiedts{}, then $\Exp[\hat{F}] = F$. Moreover, given any $\eps>0, \delta>0$ and number of samples $s=3 w(\bS) B\ln{(2/\delta)}/\eps^2F$, then with probability at least $1-\delta$ (this probability is over the randomness of \modifiedts{}; there is no stochastic assumption on $G$), $|\hat{F} - F| \leq \eps F$.

Let $\bS$ denote the \hc{} Tur\'an shadow of $G$ and $\sz{\bS}=\sum_{(S,\ell) \in \bS} |S|$. The running time of \modifiedts{} is $O(\alpha\sz{\bS} + sT_f + m + n)$ and
 the total storage is $O(\sz{\bS} + m + n)$.
 \end{theorem}
\begin{proof}
 The $X_i$ are all iid random variables and by the arguments in \Clm{uar}, every \hc{} in $G$ has the same probability of being returned by \sample{}. $\Exp[X_i]=\sum\limits_{K \in C_k(G)}\frac{f(K)}{w(\bS)}=\frac{F}{w(\bS)}$. Suppose $X_i \in [0,B]$. By \Thm{chernoff}, $Pr[|\sum_{i=1}^sX_i-s\Exp[X_i]| \geq \eps s \Exp[X_i] \leq \delta$ when $s=3 w(\bS) B\ln{(2/\delta)}/\eps^2F$.


 The running time and storage required are a direct consequence of the running time and storage required for \ts{} (Theorem 5.4 in ~\cite{JaSe17}. The only difference is the addition of $sT_f$ in the running time which is the time required to obtain $f$ for $s$ samples. 

 \end{proof}



\subsection{\mainalg{}}


 We observed that with \ts{}, bulk of the time is spent in building the tree, and only a small fraction is needed for sampling. To give a few examples, for the web-Stanford graph, construction of the shadow took 155 seconds for approximating number of 7 cliques, while taking $50K$ samples required 0.2 seconds. Similar results were observed for all other graphs we experimented with. Thus, naturally, to optimize the performance of \ts{} it would be beneficial to minimize the fraction of the shadow that is required to be built.
Consider one extreme of minimizing building the shadow - we will call it level 1 sampling. Let $N_v^+$ be the outneighborhood of $v$ in $DG$, $\Phi_v = {|N_v^+| \choose h-1}$ and $\Phi = \sum\limits_v{\Phi_v}$.  $\{(v, N^+_v, h-1)\}$ is an $h$-clique \ps{} of $G$. If we sample a $v$ with probability proportional to $\Phi_v$, and sample $h-1$-tuple of vertices from $N^+_v$ u.a.r., the probability of sampling a particular $h-1$-clique in $N^+_v$ would be $\Phi_v/\Phi * 1/\Phi_v = 1/\Phi$. If there are $C_h$ $h$-cliques in $G$ then the probability that a sampled set of $h-$vertices is a clique is $C_h/\Phi$ (we call this the success ratio). Hence, number of samples required to find a $h$-clique would be $O(\Phi/C_h)$. But $\Phi$ is typically very large compared to $C_h$ and hence the number of samples required would be very large. In other words, most of the $h-$vertex sets picked will not be cliques. 

\ts{} remedies this by first finding the Tur\'an shadow and then sampling within the subgraphs of the shadow which are dense and hence require lesser samples to find a $k$-clique. Thus, \ts{} saves on the number of samples required at the cost of building the shadow. 

The advantage of level 1 sampling is that we do not need to spend time finding the Tur\'an Shadow. We mimic the process of sampling an $h$-clique from this \ps{}, but boost the success ratio by using the latter approach. In particular, we sample a $v$ proportional to $\Phi_v={|N_v^+| \choose h-1}$, and obtain the $h-1$-clique \pts{} $\bS$ of $N_v^+$. Suppose the shadow size $\phi_v=\sum_{(P,S,\ell) \in \bS}{{S \choose \ell}}$ then probability of sampling a $h-1$-clique $=C_{h-1}(G|_{N_v^+})/\phi_v$. Thus, the success ratio goes from $C_{h-1}(G|_{N_v^+})/\Phi_v$ to $C_{h-1}(G|_{N_v^+})/\phi_v$. Since $\phi_v$ is typically much smaller than $\Phi_v$, the success ratio is much improved. However, to account for the fact that we are now sampling u.a.r. in a search space of size $\phi_v$ and not $\Phi_v$, we give a smaller weight $(\phi_v/\Phi_v)$ to every clique obtained from $N_v^+$.

 \begin{algorithm} 
\caption{\mainalg$(G, h, s, Func)$ }
\label{alg:mainalg}
Order $G$ by degeneracy and convert it to a DAG $DG$. \\
Let $M$ be a map, $W=0$ \\
Set probability distribution $D$ over $V$ where $p(v) = \sum\limits_v{\Phi_v / \Phi}$. \\
For $i = 1,2,...,s$: \\
\ \ \ \ Independently sample a vertex $v$ from $D$. \\ \label{step:sample-v}
\ \ \ \ If $M[v]$ exists, set $\bS=M[v]$ \\
\ \ \ \ else \\
\ \ \ \ \ \ \ \ $\bf{S}$ $= \ptsfinder(G_{|N_v^+}, h-1)$\\ \label{step:find-pts}
\ \ \ \ \ \ \ \ $M[v]=\bf{S}$ \\
\ \ \ \ Let $\phi_v = \sum\limits_{(P,S,\ell) \in \bf{S}}{|S| \choose \ell}$ \\
\ \ \ \ Let $K = \{v\} \cup \sample(\bf{S})$ \\ \label{step:sample-k}
\ \ \ \ If $K$ is a clique, set $X_i = \frac{\phi_v}{\Phi_v} *Func(G, K)$ \\ \label{step:func}
\ \ \ \ else set $X_i = 0$ \\
\ \ \ \ $W = W + X_i$ \\
let $\hat{F} = \frac{W}{s}\Phi$ \\
return $\hat{F} $
\end{algorithm}

For an element $E=(P,S,\ell) \in \bS $ where $\bS$ is the $k-$clique \pts{} of a graph $G$, let $\cK(E)=\{P \cup c, c \in C_{\ell}(S)\}$ denote the set of \kcs{} obtained from $E$.

\begin{lemma} \label{lem:exp}
Let $\hat{F}$ be the value returned by \mainalg{}. Then $\Exp[\hat{F}] = F$.
 \end{lemma}
 \begin{proof}
Consider an \hc{} $K \in C_h(G)$ and let $v$ be the lowest order vertex according to degenerecy ordering of vertices in $G$. Let $E = (v, N^+_v,h-1)$ then, $K \in \cK(E)$. 

Let $\bS_v$ be the $h-1$-clique \pts{} of $G|_{N^+_v}$ and let $E_v=(P,S,\ell) $ be the element in $\bS_v$ such that $K = {v} \cup P \cup c, c \in C_{\ell}(S)$.

$Pr(K \text{ is sampled in \Step{sample-k}})=$ $Pr(E \text{ is sampled}) * Pr(E_v \text{ is sampled}) * Pr(c \text{ is sampled})$ $= \frac{\Phi_v}{\Phi} * \frac{{|S| \choose \ell}}{\phi_v} * \frac{1}{{|S| \choose \ell}} = \frac{\Phi_v}{\Phi \phi_v}$
Thus, $\Exp[X_i] = \sum\limits_{v \in V}{\sum\limits_{K \in \cK(E_v)}{\frac{\Phi_v}{\Phi \phi_v}\frac{\phi_v}{\Phi_v}f(K)}}=\sum\limits_{K \in C_k(G)}{\frac{f(K)}{\Phi}}=\frac{F}{\Phi}$

Moreover, $W = \sum\limits_{i = 1}^{s}{X_i}$. Therefore, $\Exp[W]=\Exp[\sum\limits_{i = 1}^{s}{X_i}]=\sum\limits_{i = 1}^{s}{\Exp[X_i]}=s\frac{F}{\Phi}$.

Hence, $\Exp[\hat{F}]=\Exp[\frac{W}{s}{\Phi}]=F$.
  \end{proof}
\begin{theorem}\label{thm:main}
Let $f$ be a function over $h$-cliques, bounded above by $B$ such that given an $h$-clique, it takes $O(T_f)$ time to obtain the value of $f$. Given any $\eps>0, \delta>0$ and number of samples $s=3 \Phi B\ln{(2/\delta)}/\eps^2F$, \mainalg{} outputs an estimate $\hat{F}$ such that with probability at least $1-\delta$, $|\hat{F} - F| \leq \eps F$.

 Let $\bS$ denote the \kc{} Tur\'an shadow of $G$ and $\sz{\bS}=\sum_{(S,\ell) \in \bS} |S|$. The running time of \mainalg{} is $O(min(s\alpha^{h},\alpha\sz{\bS}) + sT_f + m + n)$ and
 the total storage is $O(\sz{\bS} + m + n)$.
 \end{theorem}
 \begin{proof}
 The $X_i$ are all iid random variables and by the arguments in \Lem{exp}, their expectation $\mu=F/\Phi$. Suppose $X_i \in [0,B]$. By \Thm{chernoff}, $Pr[|\sum_{i=1}^sX_i-\mu s| \geq \eps s \mu] \leq \delta$ when $s=3 \Phi B\ln{(2/\delta)}/\eps^2F$.


 The degeneracy of $G$ can be computed in time linear in the size of the graph~\cite{MB83}.
 For any $v$, the map $M[v]$ in \mainalg{} stores the $h-1$-clique \pts{} of $N_v^+$. For any $v$ that gets sampled in \Step{sample-v}, \mainalg{} checks if the \pts{} of $N_v^+$ has been constructed and if so, it uses the already-constructed shadow. If not, it constructs it in \Step{find-pts} and stores it in $M$. Thus, in the worst case, it calculates the \pts{} of $N_v^+$ for every $v$ i.e. it calculates the \pts{} of $G$ which requires time $O(\alpha\sz{\bS})$ according to Thm. 5.4 from ~\cite{JaSe17}. On the other hand, given any $v$, the size of $N_v^+$ is atmost $\alpha$ so constructing the $h-1$-clique \pts{} takes time at most $O(\alpha^{h})$ (\Clm{ptsf}) and it samples $s$ such vertices from $D$ so time required is $O(s\alpha^h)$.

 There are $s$ $h$-vertex sets sampled in \Step{sample-k} and checking if the sampled vertices form a clique takes time $h^2$, while calculating $f$ given that the sampled set is a clique, takes time $T_f$.

 Thus, the total time required by \mainalg{} is $O(min(\alpha\sz{\bS}, s\alpha^k)+ sT_f + m + n)$. 

 \end{proof}

 Depending on which structure we are counting, we can find appropriate values for $B$ and $T_f$. Notice that in the worst case, depending on the structure of the graph, \mainalg{} may end up building the entire shadow in which case it will not provide any savings over \modifiedts{}. However, practically, we observe that we get significant savings in the amount of shadow built using \mainalg{} in most cases. Unless specified otherwise, all results in this paper are obtained using \mainalg{}.




\section{Counting cliques and near-k-cliques} \label{sec:counters}





\subsection{Counting \kocs{}}

\begin{algorithm}
\caption{\funckoc$(G,K)$
}
\label{alg:funckoc}
$f'=0$ \\
Let $u$ and $v$ be two distinct vertices from $K$ \\
Let $nbrs = N_u \cup N_v$ \\
For $nbr \in nbrs$: \\
\ \ \ \ If $nbr$ is connected to all vertices in $K$ except 1 vertex, say $w$ and $nbr>w$, then $f' = f'+1$\\ 
return $f'$
\end{algorithm}

\begin{definition}
Let $(u,v), u<v$, be the missing edge in a \koc{} $J$. The lower-order $k-1$-clique in $J$ is the $k-1$-clique $J \setminus \{u\}$, and $J \setminus \{v\}$ is the higher-order $k-1$-clique in $J$. 
\end{definition}

\begin{claim} \label{clm:kocs}
Let $f(K)$ for $k-1$-clique $K$ denote the number of \kocs{} that $K$ is the lower-order $k-1$-clique in. Then $F= \sum\limits_{K \in C_{k-1}(G)}{f(K)}=$ total number of \kocs{} in $G$.
\end{claim}
\begin{proof}
Every \koc{} has exactly $1$ lower-order $k-1$-clique. If $f(K)$ denotes the number of \kocs{} that $K$ is a part of and is the lower-order clique in, then $\sum\limits_{K \in C_{k-1}(G)}{f(K)}=F=$ total number of \kocs{} in $G$.
\end{proof}
\begin{claim}
For input $k-1$-clique $K$, \funckoc{} returns $f(K)$.
\end{claim}
\begin{proof}
For any $nbr\in V$, if $K \cup \{nbr\}$ is a \koc, then either $nbr \in N_u$ or $nbr \in N_v$ or both.
 For a given $K$, \funckoc{} finds the set of $nbr$ ($nbrs$) that are connected to every vertex in $K$ except one. Thus, every $\{nbr\} \cup K$ for $nbr \in nbrs$ is a \koc{} and it is counted in $f'$ iff $K$ is a lower-order $k-1$-clique. Thus, the value returned, $f'=f(K)$.
\end{proof}



\begin{theorem}
Let $d_{max}$ be the maximum degree of any vertex in $G$. Then $B=min(2d_{max}, n)$ and $T_f=O(d_{max})$ for \funckoc{}.

\end{theorem}

\begin{proof}
By \Clm{kocs}, $F=$total number of \kocs{} in $G$. For any \koc{} $J = K \cup \{nbr\}$ that $K$ is the lower-order $k-1$-clique in, either $nbr \in N_u$ or $nbr \in N_v$ or both. Thus the number of \kocs{} in which it is the lower-order $k-1$-clique is atmost $2d_{max}$. On the other hand, there can be atmost $n$ $nbr$, thus $B=min(2d_{max}, n)$. 
Finding $nbrs$ takes time $O(d_{max})$ and checking if $nbr \in nbrs$ forms a \koc{} with $K$ takes time $O(1)$. Hence, $T_f=O(d_{max})$

\end{proof}

\subsection{Counting Type 1, \ktcs{}}
\begin{algorithm}
\caption{\funcktc$(G,K)$ 
}
\label{alg:funcktc}
$f'=0$ \\
For $u \in K$: \\
\ \ \ \ For $v \in K, v>u$: \\
\ \ \ \ \ \ \ \ Let $nbrs$ be the set of vertices connected to all vertices in $K$ except $u$ and $v$ \\
\ \ \ \ \ \ \ \ $f' = f' + |nbrs|$ \\
return $f'$
\end{algorithm}

\begin{claim}
Let $f(K)$ for $k-1$-clique $K$ denote the number of Type 1 \ktcs{} that $K$ is contained in. Then $F=\sum\limits_{K' \in C_{k-1}(G)}{f(K')}=$ the total number of Type 1 \ktcs{} in $G$.
\end{claim}
\begin{proof}
Every Type 1 \ktc{} contains exactly $1$ $k-1$-clique (\Fig{types-annotated}). Thus, $\sum\limits_{K' \in C_{k-1}(G)}{f(K')}=F=$ the total number of Type 1 \ktcs{} in $G$.
\end{proof}

\begin{claim}
For input $k-1$-clique $K$, \funcktc{} returns $f(K)$.
\end{claim}
\begin{proof}
Given $K$, for every distinct pair of vertices $u$ and $v \in K, v > u$, \funcktc{} finds the set of vertices $nbrs$ such that $\forall nbr \in nbrs$, $nbr$ is connected to all vertices in $K$ except $u$ and $v$. Thus, $K \cup \{nbr\}$ is a \kc{} with exactly 2 edges missing - $(u,nbr)$ and $(v,nbr)$ with the missing edges having a vertex in common $(nbr)$ i.e. it is a Type 1 \ktc{}. 
Thus, \funcktc{} returns the number of Type 1 \ktcs{} that $K$ is contained in i.e. it returns $f(K)$.
\end{proof}

\begin{theorem}
$B=min(3d_{max},n)$, $T_f = O(d_{max})$ for \funcktc{}.

\end{theorem}
\begin{proof}
For any $3$ vertices $u,v,w \in K$ and for any \ktc{} $J = K \cup \{nbr\}$ that $K$ is contained in, atleast one of $(u,nbr), (v,nbr), (w,nbr) \in E(G)$. Thus, any $K$ can be a part of atmost $min(3d_{max},n)$ Type 1 \ktcs{}. For every pair $(u,v)$ in $K$, \funcktc{} calculates the number of vertices connected to all in $K$ but $u$ and $v$ which takes time $O(d_{max})$. Thus, $T_f = O(d_{max})$. 
\end{proof}

\subsection{Counting Type 2 \ktcs{}}

\begin{algorithm}
\caption{\funcktcc$(G,K)$ 
}
\label{alg:funcktcc}
$f'=0$ \\
Let $degen(u)$ denote the position of $u$ in the degeneracy order of $G$. \\
For $u \in K$: \\ \label{step:loopu}
\ \ \ \ For $w \in K, degen(w)>degen(u)$: \\ \label{step:loopw}
\ \ \ \ \ \ \ \ Let $nbrsu = N^+_u$ be the set of out-nbrs of $u$ such that they are connected to all vertices in $K$ except $w$ and $\forall nbru \in nbrsu, degen(w) < degen(nbru)$.\\ \label{step:conditionu}
\ \ \ \ \ \ \ \ Let $nbrsw$ be the set of neighbors of $w$ in $G$ such that they are connected to all vertices in $K$ except $u$ \\ \label{step:conditionw}
\ \ \ \ \ \ \ \ For $x \in nbrsu$: \\
\ \ \ \ \ \ \ \ \ \ \ \ For $v \in nbrsw$:\\
\ \ \ \ \ \ \ \ \ \ \ \ \ \ \ \ If $(nbru, nbrw) \in E(G): f'=f'+1$ \\ \label{step:xv}
return $f'$
\end{algorithm}

\begin{definition}
Given a Type 2 \ktc{} $J$, $v,x \in J$, the set $K=J \setminus \{v,x\}$ is the lowest order $k-2$-clique of $J$ if it fulfills all the following conditions:
\begin{enumerate}
\item $(u,v) \notin E(G)$, $(w,x) \notin E(G)$ (note that this implies that $K$ is a $k-2$-clique). \label{prop:missing}
\item $degen(u)<degen(v)$ \label{prop:umin}
\item $degen(u)<degen(w)<degen(x)$.\label{prop:degen-order}
\end{enumerate}
\end{definition}
Note that $u,v,w$ and $x$ are all distinct and $J$ consists of exactly 4, $k-2$-cliques: $J \setminus \{v,x\}$, $J \setminus \{v,w\}$, $J \setminus \{u,x\}$ and $J \setminus \{u,w\}$ (\Fig{types-annotated}), and the lowest order $k-2$-clique of $J$ is the one which has the vertex $(u)$ with minimum position in the degeneracy ordering of $G$ and the minimum neighbor of $u$.
\begin{claim} \label{clm:ktcs2}
Let $f(K)$ for $k-2$-clique $K$ denote the number of Type 2 \ktcs{} that $K$ is the lowest-order $k-2$-clique in. Then $F=\sum\limits_{K' \in C_{k-2}(G)}{f(K')}=$ total number of Type 2 \ktcs{} in $G$.
\end{claim}
\begin{proof}
Every Type 2 \ktc{} has exactly one lowest order $k-2$-clique in it. If $f(K)$ denotes the number of Type 2 \ktcs{} that $K$ is the lowest-order $k-2$-clique in, then $\sum\limits_{K' \in C_{k-2}(G)}{f(K')}=$total number of Type 2 \ktcs{} in $G$.
\end{proof}

\begin{claim}
For input $k-2$-clique $K$, \funcktcc{} returns $f(K)$.
\end{claim}
\begin{proof}

Given a $k-2$-clique $K$, 
\Step{loopu} and \Step{loopw} loop over all possible candidates for $u$ and $w$, maintaining the condition that $degen(u)<degen(w)$. In \Step{conditionu}, \funcktcc{} picks the outneighbors of $u$ that are potential candidates for $x$ $(nbrsu)$ such that $(w,x) \notin E(G)$ and $degen(w)<degen(x)$. In \Step{conditionw}, it picks potential candidates for $v$ $(nbrsw)$ i.e. neighbors of $w$ that are connected to all vertices in $K$ except $u$. Finally, in \Step{xv}, it checks if $v$ and $x$ are connected. Thus, $f'$ in \Step{xv} is incremented iff all the conditions of a lowest order $k-2$-clique of a Type 2 \ktc{} are fulfilled. Thus, the returned value $f'=f(K)$.
\end{proof}

\begin{theorem}
$B=min(n^2, k^2\alpha d_{max}/2)$, $T_f = O(\alpha+d_{max})$ for \funcktcc{}.

\end{theorem}
\begin{proof}
Given $K$, there can be atmost $k^2/2$ candidates for $(u,w)$. There can be at most $\alpha$ candidates for $x$ (since it has to be an outneighbor of $u$) and atmost $d_{max}$ candidates for $v$ (neighbors of $w$). On the other hand, there can be atmost $n$ candidates for $x$ and $v$ each. Thus, $B=min(n^2, k^2\alpha d_{max}/2)$

Given a set of $k-2$ vertices, it takes $O(k^2)$ time to check if it forms a clique. There are $O(k^2)$ candidates for $(u,w)$ each. There are atmost $\alpha$ candidates for $x$ and $d_{max}$ candidates for $v$ whose connections to each of the $k-2$ vertices need to be checked. This takes time $O(\alpha+d_{max})$. Altogether, $T_f = O(\alpha+d_{max})$. 
\end{proof}

\section{Experimental Results} \label{sec:results}
\textbf{Preliminaries:} We implemented our algorithms in {\tt C++} and ran our experiments on a
commodity machine equipped with a 1.4GHz AMD Opteron(TM) processor 6272 with 8~cores
and  2048KB  L2 cache (per core), 6144KB L3 cache, and  128GB memory. We performed our experiments on a collection of graphs from
SNAP~\cite{Snap}, including social networks, web networks, and infrastructure networks.
 The largest graph
has more than 100M edges. Basic properties like degneracy, maximum degree etc. of these graphs are presented in Table~\ref{tab:GraphStats}. 
We consider the graph to be simple and undirected. Code for all experiments is available at: \href{https://bitbucket.org/sjain12/counting-near-cliques}{https://bitbucket.org/sjain12/counting-near-cliques}

Our practical implementation differs slightly from \mainalg{} in two ways: we fix the number of samples to 500K. Moreover, since the number of samples are fixed, we can sample from $D$ in \mainalg{} all at once and maintain counts of the number of cliques to be sampled from each outneighborhood. We can then explore the outneighborhoods in an online fashion, sampling as we build the shadow. Once the samples from a vertex's outneighborhood have been ontained, we no longer need the shadow of the outneighborhood and the shadow can be discarded. Thus, we don't need to store the entire shadow but only the shadow of the current vertex's outneighborhood. 

We focus on counting near-$k$-cliques for $k$ ranging from $5$ to $10$.

\begin{table*}[]
\scriptsize
\centering
\begin{tabular}{|l|l|l|r|r|r|l|r|r|l|r|r|l|r|r|}
\hline

\multicolumn{5}{|c|}{}  & \multicolumn{3}{c|}{\textbf{k=5}}                       & \multicolumn{3}{c|}{\textbf{k=7}}                       & \multicolumn{3}{c|}{\textbf{k=10}}      &              \\ \cline{7-7}
\hline
\textbf{graph}         & \textbf{vertices}     & \textbf{edges}        & \textbf{degen}        & \textbf{$d_{max}$}  &\textbf{estimate} & \textbf{\% error} & \textbf{time} & \textbf{estimate} & \textbf{\% error} & \textbf{time} & \textbf{estimate} & \textbf{\% error} & \textbf{time}  & \textbf{type} \\ \cline{7-7}
\hline
\multirow{4}{*}{web-Stanford} &\multirow{4}{*}{2.82E+05}              &\multirow{4}{*}{1.99E+06}              &\multirow{4}{*}{71}                    &\multirow{4}{*}{38625}           & 2.36E+10          & 0.85                 & 142            & 8.99E+11          & -                 & 216            & 2.16E+14          & -                 & 129     & $(k,1)$        \\
& & & & & 1.15E+11          & 0.46                 & 8283            & 7.33E+11          & -                 & 3802            & 1.12E+14          & -                 & 1087      & $(k,2)$ Type 1     \\
& & & &  &1.12E+10          & 1.19                & 5396            & 2.51E+11          & -                 & 538            & 1.04E+14          & -                 & 293     & $(k,2)$ Type 2    \\
& & & &  & 6.21E+8          &                  &             & 3.47E+10          &                 &             & 5.82E+12          &                 &     & $k$     \\
\hline
\multirow{4}{*}{web-Google}             & \multirow{4}{*}{8.76E+05}              & \multirow{4}{*}{4.32E+06}              & \multirow{4}{*}{44}                    & \multirow{4}{*}{6332}                  & 6.76E+08          & 0.44              & 13             & 2.19E+09          & 0.45              & 12             & 2.41E+10          & 0.41              & 10            & $(k,1)$  \\
& & & & & 2.08E+09          & 0.48              & 276             & 4.45E+09          & 0.01              & 172             & 2.05E+10          & -              & 42          & $(k,2)$ Type 1   \\
& & & & & 7.18E+07          & 1.10              & 21             & 2.93E+08          & 0.01              & 18            & 7.70E+09          & 0.01              & 13          & $(k,2)$ Type 2   \\
& & & & & 1.05E+08          &               &              & 6.06E+08          &              &            & 1.29E+10          &             &          & $k$   \\
\hline
\multirow{4}{*}{amazon0601}             &\multirow{4}{*}{4.03E+05}              & \multirow{4}{*}{4.89E+06}              & \multirow{4}{*}{10}                    & \multirow{4}{*}{2752}                  & 1.17E+07          & 0.00              & 4             & 2.88E+06          & 0.01              & 3             & 3.76E+04          & 0.02              & 1.5            & $(k,1)$ \\
& & & & & 5.38E+07          & 0.01              & 10             & 7.84E+06          & 0.01              & 7             & 8.70E+04          & 0.01              & 3          & $(k,2)$ Type 1   \\
& & & & & 3.16E+06          & 0.01              & 4             & 1.30E+06          & 0.01              & 5             & 2.96E+04          & 0.00              & 3          & $(k,2)$ Type 2   \\
& & & & & 3.64E+06          &               &              & 9.98E+05          &           &             & 9.77E+03          &              &           & $k$   \\
\hline
\multirow{4}{*}{web-BerkStan}           &\multirow{4}{*}{6.85E+05}              & \multirow{4}{*}{6.65E+06}              & \multirow{4}{*}{201}                   & \multirow{4}{*}{84230}                 & 4.89E+11          & 0.93                 & 397           & 2.89E+13          & -                 & 470           & 1.85E+16          & -                 & 704          & $(k,1)$ \\
& & & & & 1.89E+12          & 0.32                 & 20534           & 7.39E+13          & -                 & 6080           & 1.43E+16          & -                 & 5383         &  $(k,2)$ Type 1 \\
& & & & & 6.61E+10          & 0.09                 & 12400           & 7.32E+11          & -                 & 605           & 1.65E+14          & -                 & 646          & $(k,2)$ Type 2 \\
& & & & & 2.19E+10          &                  &           & 9.30E+12          &                 &            & 5.79E+16          &                &        & $k$ \\
\hline
\multirow{4}{*}{as-skitter}             & \multirow{4}{*}{1.70E+06}              & \multirow{4}{*}{1.11E+07}              & \multirow{4}{*}{111}                   & \multirow{4}{*}{35455}                 & 3.94E+10          & 4.52                 & 1180           & 5.44E+11          & -                 & 1034           & 7.91E+13          & -                 & 800          & $(k,1)$ \\
& & & & & 2.97E+11          & 1.63                & 31724           & 2.48E+12          & -                 & 16220           & 2.27E+13          & -                 & 10461         & $(k,2)$ Type 1  \\
& & & & & 2.34E+10          & 1.37                & 4132           & 3.97E+11          & -                 & 2598           & 8.55E+13          & -                 & 1038         & $(k,2)$ Type 2  \\
& & & & & 1.17E+09          &                 &           & 7.30E+10          &               &            & 1.43E+13          &                &          & $k$  \\
\hline
\multirow{4}{*}{cit-Patents}            & \multirow{4}{*}{3.77E+06}              & \multirow{4}{*}{1.65E+07}             & \multirow{4}{*}{64}                    & \multirow{4}{*}{793}                   & 4.12E+07          & 0.01              & 10            & 7.20E+07          & 0.01              & 6             & 9.06E+05*          & 42.22              & 4            & $(k,1)$ \\
& & & & & 1.11E+08          & 1.83              & 17            & 1.31E+08          & 2.29              & 8             & 1.43E+06*          & 49.11              & 5          &  $(k,2)$ Type 1  \\
& & & & & 1.31E+08          & 0.01              & 6            & 6.76E+08          & 3.36             & 9             & 2.54E+07*          & 31.35              & 5           & $(k,2)$ Type 2  \\
& & & & & 3.05E+06          &               &             & 1.89E+06          &            &            & 2.55E+03          &              &            & $k$  \\
\hline
\multirow{4}{*}{soc-pokec}              & \multirow{4}{*}{1.63E+06}              & \multirow{4}{*}{2.23E+07}              & \multirow{4}{*}{47}                    & \multirow{4}{*}{14854}                 & 4.22E+08*          & 8.48              & 218            & 5.41E+07*          & 9.96              & 81            & 7.67E+08          & 4.24              & 55           & $(k,1)$ \\
& & & & & 2.40E+09*          & 6.19              & 218            & 1.59E+09*          & 4.6              & 136            & 1.67E+09          & 0.02              & 68        & $(k,2)$ Type 1   \\
& & & & & 3.34E+08          & 0.00              & 38            & 6.78E+08*          & 7.61              & 95            & 1.28E+09          & 0.01              & 64          & $(k,2)$ Type 2  \\
& & & & & 5.29E+07          &               &          & 8.43E+07          &         &          & 1.98E+08          &              &           & $k$  \\

\hline
\multirow{4}{*}{com-lj}                 & \multirow{4}{*}{4.00E+06}              & \multirow{4}{*}{3.47E+07}              & \multirow{4}{*}{360}                   & \multirow{4}{*}{14815}                 & 2.85E+11          & 0.11                 & 200           & 4.28E+14          & -                 & 452           & 1.18E+19          & -                 & 558      & $(k,1)$     \\
& & & & & 4.63E+11          & 0.34                 & 756           & 5.11E+14          & -                 & 613           & 1.22E+19          & -                 & 680          & $(k,2)$ Type 1  \\
& & & & & 5.39E+10          & 0.53                 & 269           & 1.24E+14          & -                 & 581           & 4.23E+18          & -                 & 568          & $(k,2)$ Type 2 \\
& & & & & 2.47E+11          &               &          & 4.51E+14          &            &           & 1.47E+19          &                 &          & $k$ \\
\hline
\multirow{4}{*}{soc-LJ}            & \multirow{4}{*}{4.84E+06}              & \multirow{4}{*}{8.57E+07}              & \multirow{4}{*}{372}                    & \multirow{4}{*}{20333}                 & 6.32E+11          & 0.03              & 677             & 1.01E+15          & -              & 779             & 4.14E+19          & -              & 960            & $(k,1)$ \\
& & & & & 1.03E+12          & 0.17              & 1504             & 1.27E+15          & -             & 1107             & 4.57E+19          & -              & 1320           & $(k,2)$ Type 1   \\
& & & & & 1.34E+11          & 0.41              & 506             & 2.77E+14          & -              & 1007             & 1.17E+19          & -              & 1111          & $(k,2)$ Type 2   \\
& & & & &           &               &              & 4.49E+14          &               &             &           &              &          & $k$   \\
\hline
\multirow{4}{*}{com-orkut}              & \multirow{4}{*}{3.07E+06}              & \multirow{4}{*}{1.17E+08}              & \multirow{4}{*}{253}                   & \multirow{4}{*}{33313}                 & 1.56E+11          & -                & 9507          & 2.26E+12          & -                 & 16546          & 4.66E+13          & -                 & 26370      & $(k,1)$ \\  
& & & & & 1.46E+12          & -                 & 21213          & 7.82E+12          & -                 & 24148          & 1.04E+14          & -                 & 29881      & $(k,2)$ Type 1  \\  
& & & & & 2.37E+11          & -                 & 3879          & 3.51E+12          & -                 & 11617          & 1.60E+14          & -                 & 22676     & $(k,2)$ Type 2  \\  
& & & & & 1.57E+10          &                  &        & 3.61E+11          &               &           & 3.03E+13          &                &      & $k$  \\  
\hline
\end{tabular}
\caption{Table shows the sizes, degeneracy, maximum degree of the graphs, the counts of 5, 7 and 10 cliques and near-cliques obtained using \mainalg{}, the percent relative error in the estimates (for those graphs for which we were able to get exact numbers within 24 hours), and time in seconds required to get the estimates. The rows whose types are $k$ in the rightmost column show the number of \kcs{}.For most instances, the algorithm terminated in minutes. Values marked with * have significant errors which are addressed in \Tab{CorrectedValues}}
 \label{tab:GraphStats}
\end{table*}

\textbf{Accuracy and convergence of \mainalg:}
We picked some graphs for which
the exact \nkc{} counts are known (for all $k \in [5,10]$). 
For each graph and near-clique type, for sample size in [10K,50K,100K,500K,1M], we performed 100 runs
of the algorithm. We show here results for {\tt amazon0601} for $k=7$, though similar results were observed for other graphs and $k$. We plot the spread of the output of \mainalg{}, over
all these runs. The results are shown in~\Fig{convergence}. The red line
denotes the true answer, and there is a point for the output of every
single run. As we can see, the output of \mainalg{} fast converges to the true value as we increase the number of samples. 
For 500K samples, the range of values is within $5\%$ of the true answer which is much less compared to the spread of cc. Similar results were observed for other graphs for which the exact counts were available, except soc-pokec. The error was mostly $<5\%$ and often $<1\%$ as can be seen from \Tab{GraphStats}. 

In cases like soc-pokec the error can be high. This happens when most of the samples end up empty, either because the sampled vertices did not form a clique, or the samples belonged to out-neighborhoods that did not have a clique of the required size or the sampled clique does not participate in any near-cliques. This can be detected by observing how many of the samples taken in \Step{sample-k} were cliques with non-zero $f$. If this number is $<<5000$, the estimates are likely to have substantial error. This can be remedied by either taking more samples or using \modifiedts{}. \Tab{CorrectedValues} shows the revised estimates obtained using \modifiedts{} using 500K samples, for values in \Tab{GraphStats} that have substantial error (marked with an asterisk).

For the graphs for which we could not get exact numbers (since the bf algorithm did not terminate in 1 day), we were unable to obtain error percentages. However, even for such graphs we saw good convergence over 100 runs of the algorithm.

\begin{table}[]
\scriptsize
\centering
\begin{tabular}{|l|l|c|c|c|r|}
\hline
\textbf{graph}         & \textbf{k}           & \textbf{revised estimate}        & \textbf{revised \% error} & \textbf{time}  &\textbf{type} \\ 
\hline
\multirow{3}{*}{cit-Patents} &  \multirow{3}{*}{10}  & 648944 & 1.91 & 130 & $(k,1)$        \\
&  & 2.84E+06 & 1.06 & 130 & $(k,2)$ Type 1     \\
&  & 3.69+07 & 0.27 & 130 & $(k,2)$ Type 2     \\
\hline
\multirow{2}{*}{soc-pokec} & \multirow{2}{*}{5}  & 3.91E+08 & 0.51& 284 & $(k,1)$        \\
&  & 2.27E+09& 0.44 & 371 & $(k,2)$ Type 1     \\
\hline
\multirow{3}{*}{soc-pokec} & \multirow{3}{*}{7}  & 4.92E+08 &0.01 & 288 & $(k,1)$     \\
& &  1.53E+09 & 0.24 & 347 & $(k,2)$ Type 1     \\
& &  6.27E+08 & 0.47 & 298 & $(k,2)$ Type 2     \\
\hline
\end{tabular}
\caption{Table revised estimates, revised error and time in seconds for the counts of near-cliques obtained using \modifiedts{} with 500K samples for the erroneous estimates in \Tab{GraphStats} (marked with *). }
\label{tab:CorrectedValues}
\end{table}

\begin{figure*}
\begin{subfigure}[b]{0.30\textwidth}
    \centering
    \includegraphics[width=\textwidth]{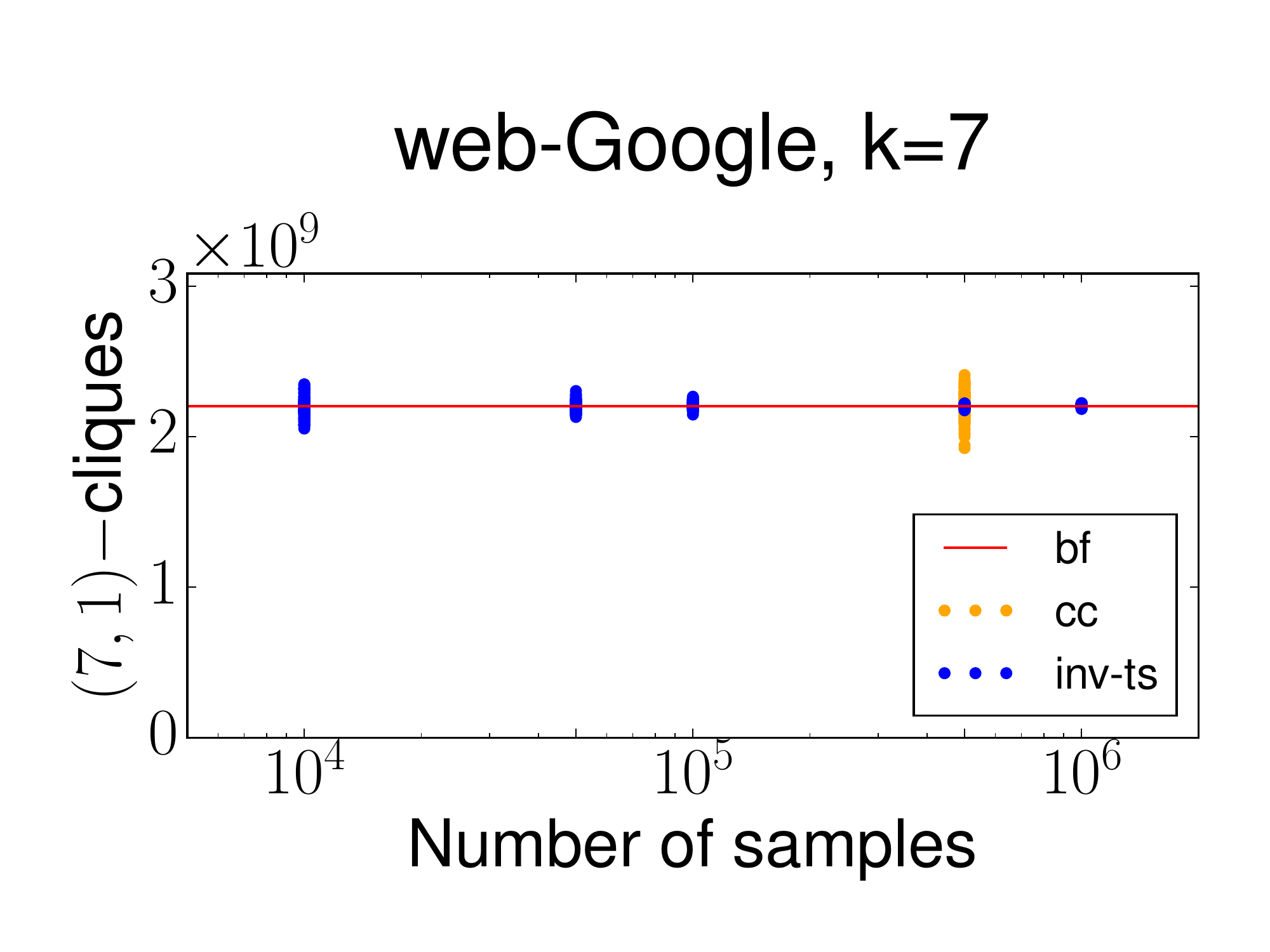}
    \caption{\koc{}}
    \label{fig:google-7-nc-convergence}
\end{subfigure}
\begin{subfigure}[b]{0.30\textwidth}
    \centering
    \includegraphics[width=\textwidth]{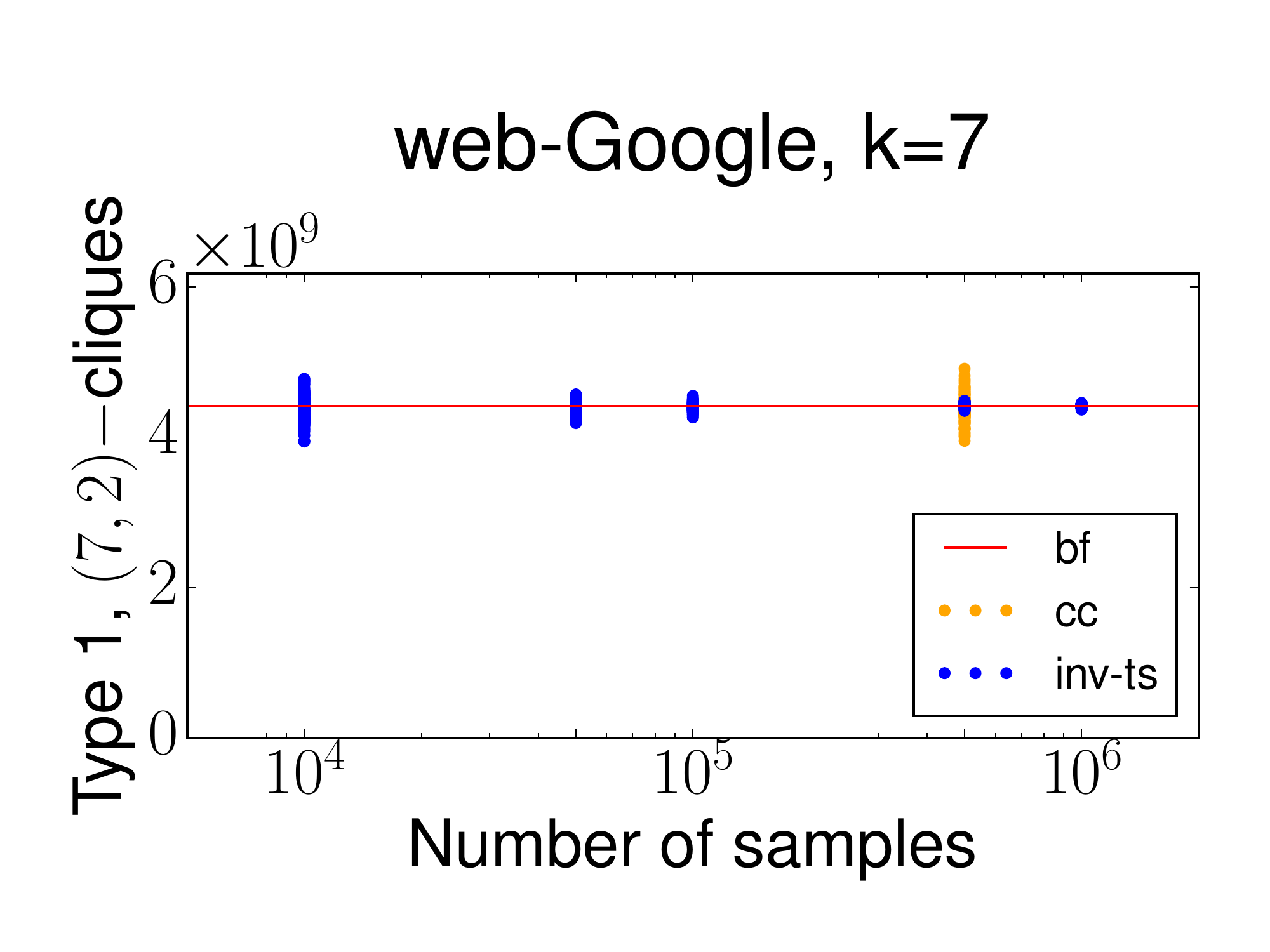}
    \caption{Type 1, \ktc{}}
    \label{fig:google-7-nc2-convergence}
\end{subfigure}
\begin{subfigure}[b]{0.30\textwidth}
    \centering
    \includegraphics[width=\textwidth]{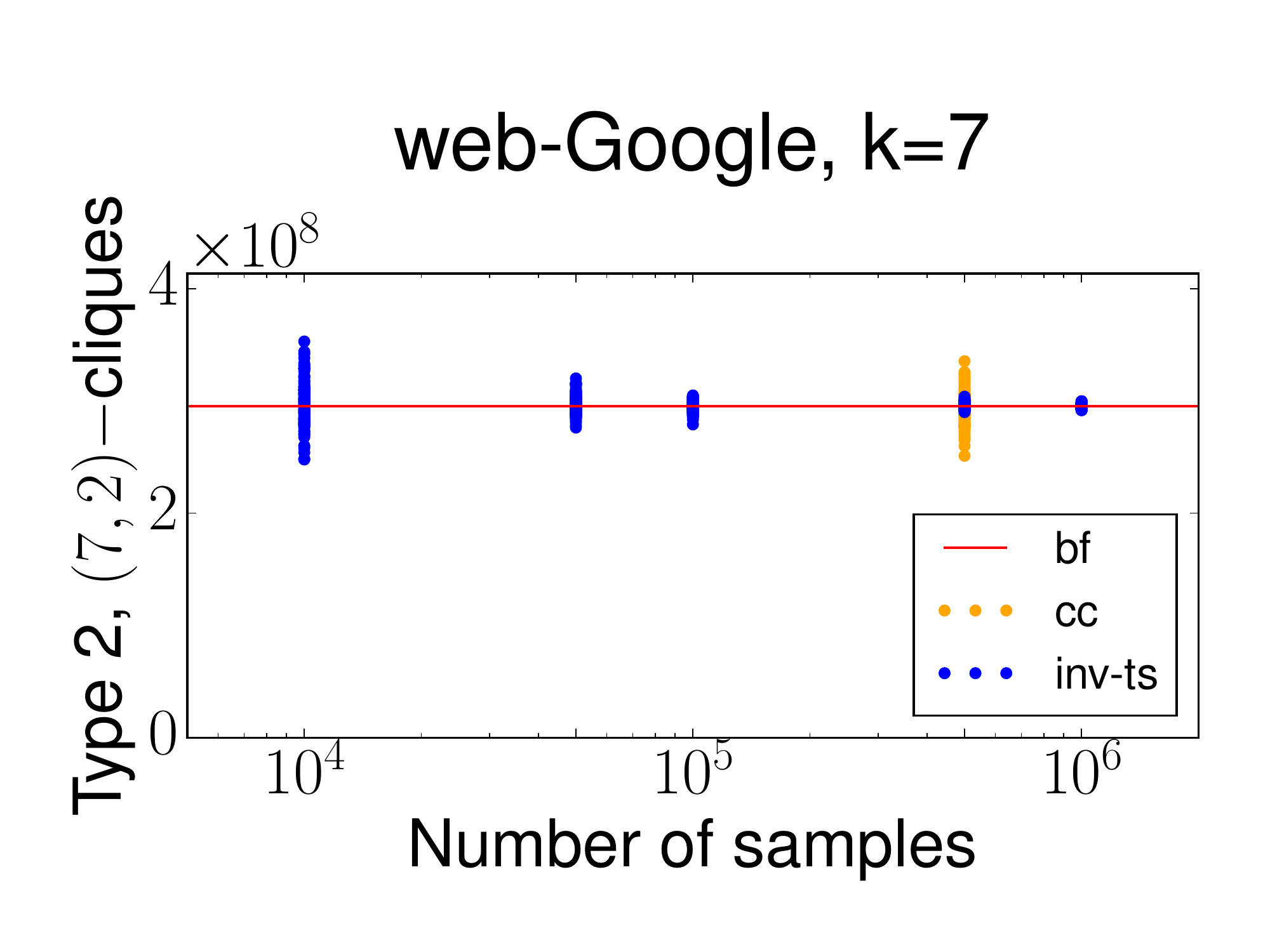}
    \caption{Type 2, \ktc{}}
    \label{fig:google-7-nc22-convergence}
\end{subfigure}
\caption{ \Fig{google-7-nc-convergence}, \Fig{google-7-nc2-convergence}, \Fig{google-7-nc22-convergence} show convergence over 100 runs of \mainalg{} using number of samples in [10K, 50K, 100K, 500K,1M] for all near-clique types. The red line indicates the true value. }

\label{fig:convergence}
\end{figure*}

\textbf{Running time:} The runtimes for near-cliques of size 7 are presented in \Tab{GraphStats}. We show the time for a single run in each case. In all cases except com-orkut, the algorithm terminated in minutes (for com-orkut, it took less than a day) where cc and bf did not terminate in an entire day (and in some cases, even after 5 days).

\begin{figure}[t]   
        \begin{subfigure}[b]{0.235\textwidth}
        \centering
        \includegraphics[width=\textwidth]{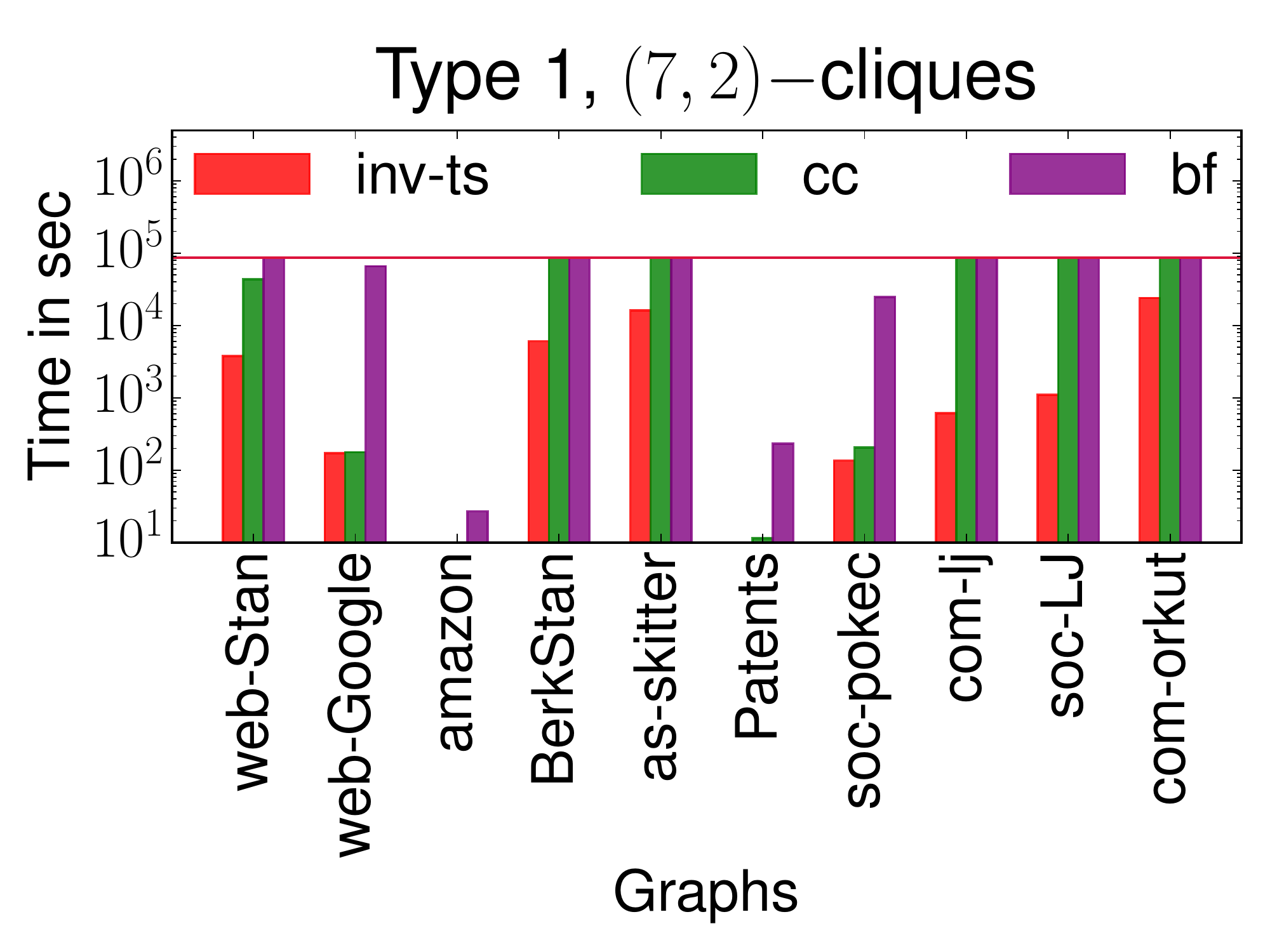}
        \label{fig:timing-7-nc2}
        \end{subfigure}
        \hfill
        \begin{subfigure}[b]{0.235\textwidth}
        \centering
        \includegraphics[width=\textwidth]{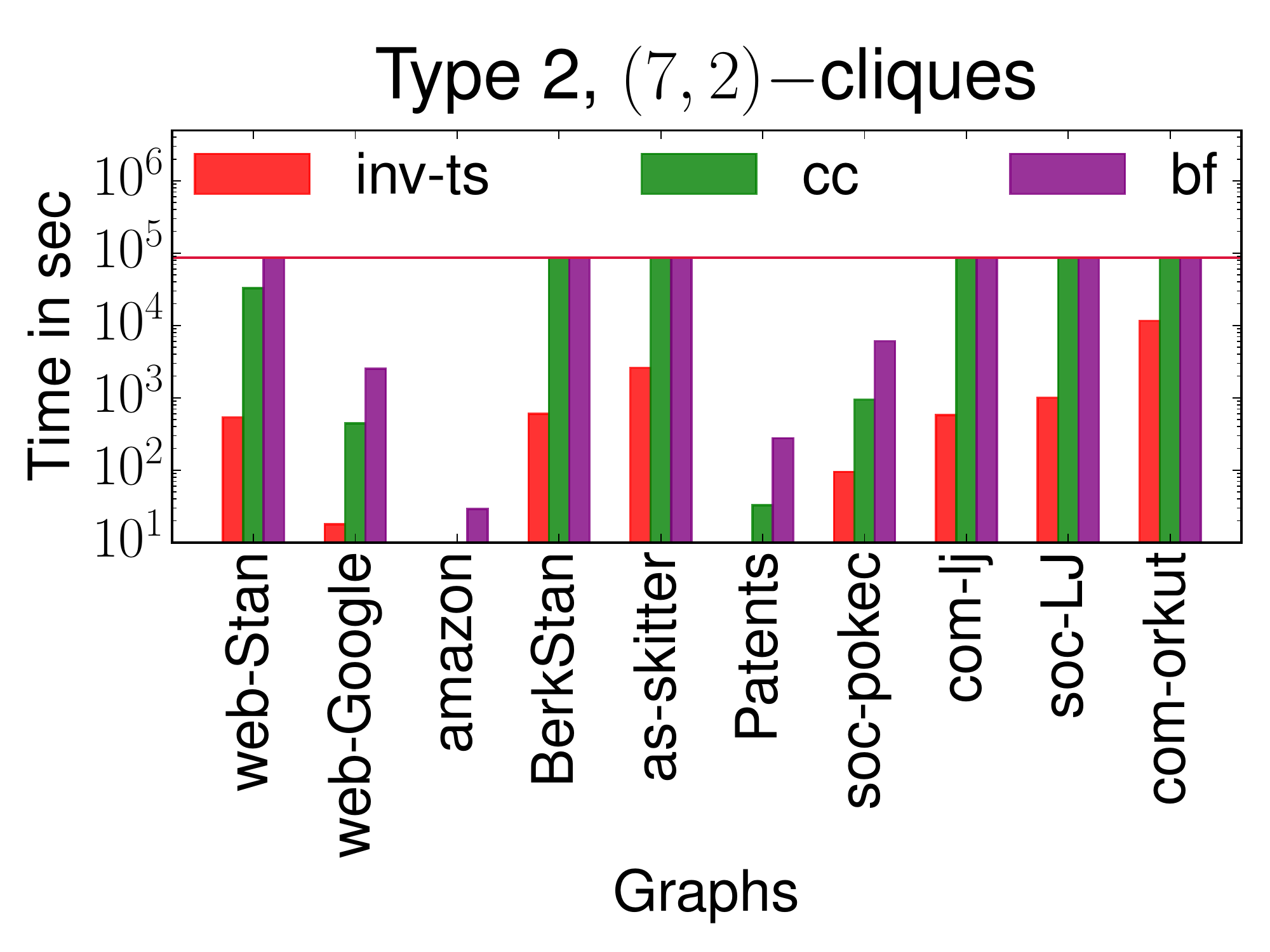}
        \label{fig:timing-7-nc22}
        \end{subfigure}
        \caption{ Figure shows the time required by \mainalg{} (inv-ts), color-coding (cc) and brute force (bf) to estimate the number of Type 1 and Type 2 \ktcs{} resp. in 10 real world graphs for $k=7$. The red line indicates 86400 seconds (24 hours). }
\end{figure}



\textbf{Comparison with other algorithms:}  
Our exact brute-force procedure is a
well-tuned algorithm that uses the degeneracy ordering and exhaustively searches outneighborhoods
for cliques (based on the approach by Chiba-Nishizeki~\cite{ChNi85}). Once a clique is found, we count all the near-cliques the clique is a part of and sum this quantity over all cliques.


On average, color-coding took time anywhere between 2x to 100x time taken by \mainalg{}, while giving poorer accuracy. Brute force took even more time. \mainalg{} has reduced the time required to obtain these estimates from days to minutes.

\subsection{Near-cliques in practice} \label{sec:practice}
One of the important applications of near-cliques is in finding missing edges that likely should have been present in the graph in the first place. We deployed our algorithm on a citation network\ ~\cite{Tang08}. Using \mainalg{} we were able to obtain several sets of papers in which, ever pair of paper either cited or was cited by the other paper (depending on the chronological order of the papers), except 1 or 2 pairs. For example, a $(7,1)$-clique we obtained comprised of the papers with the following titles: 
\begin{enumerate}
\item A ray tracing solution for diffuse interreflection 
\item Distributed ray tracing 
\item A global illumination solution for general reflectance distributions 
\item Adaptive radiosity textures for bidirectional ray tracing 
\item The rendering equation 
\item A two-pass solution to the rendering equation: A synthesis of ray tracing and radiosity methods 
\item A framework for realistic image synthesis 
\end{enumerate} 
 in which, only $(1)$ and $(3)$ were not connected. Thus, by mining near-cliques one can discover missing links and offer suggestions for which items should be related. In applications where the data is known to be noisy, it would be interesting to see how the properties of the graph change upon adding these (possibly) missing links and obtaining a more complete picture. 


 \textbf{Listing near-cliques:}
In some applications of \nkcs{}, a u.a.r. sample of \nkcs{} may be required. Suppose we want to provide a u.a.r. sample of Type 1 \ktcs{} for a given $k$. \modifiedts{} allows us to sample cliques u.a.r. Once a clique $K$ is sampled, suppose we return a u.a.r. Type 1 \ktc{} that $K$ participates in. Let $J$ be a Type 1 \ktc{} that $K$ participates in, then the probability of $J$ being returned is inversely proportional to $f(K)$. In other words, this approach does not give us a u.a.r. sample of Type 1 \ktcs{}. However, if we list all the Type 1 \ktcs{} that $K$ participates in, and repeat this process for several different $K$, even though the samples in the list may be correlated, every Type 1 \ktcs{} in $G$ has equal probability of being put in the list. In applications where some amount of correlation in samples is tolerable, such a list can be useful.

\section{Conclusion and Future Work}
We leverage the fast clique counting algorithm \ts{} to count near-cliques that are essentially $k$-cliques missing 1 or 2 edges, for $k$ upto 10. The proposed algorithm gives significant savings in space and time compared to state of the art. 

One could generalize the definition of near-cliques to larger values of $r$ and define a $(k, r)-$clique as a $k-$ clique that is missing exactly $r$ edges. It would be interesting to see how far $r$ can be increased such that near-clique counting would still be feasible using this clique-centered approach.

\begin{acks}
  Shweta Jain and C. Seshadhri acknowledge the support of \grantsponsor{}{NSF}{} Awards \grantnum{}{CCF-1740850}, \grantnum{}{CCF-1813165}, and \grantsponsor{}{ARO}{} Award \grantnum{}{W911NF1910294}.
\end{acks}

\bibliographystyle{ACM-Reference-Format}
\bibliography{WWW-mining-near-cliques}


\begin{thebibliography}{46}


\ifx \showCODEN    \undefined \def \showCODEN     #1{\unskip}     \fi
\ifx \showDOI      \undefined \def \showDOI       #1{#1}\fi
\ifx \showISBNx    \undefined \def \showISBNx     #1{\unskip}     \fi
\ifx \showISBNxiii \undefined \def \showISBNxiii  #1{\unskip}     \fi
\ifx \showISSN     \undefined \def \showISSN      #1{\unskip}     \fi
\ifx \showLCCN     \undefined \def \showLCCN      #1{\unskip}     \fi
\ifx \shownote     \undefined \def \shownote      #1{#1}          \fi
\ifx \showarticletitle \undefined \def \showarticletitle #1{#1}   \fi
\ifx \showURL      \undefined \def \showURL       {\relax}        \fi
\providecommand\bibfield[2]{#2}
\providecommand\bibinfo[2]{#2}
\providecommand\natexlab[1]{#1}
\providecommand\showeprint[2][]{arXiv:#2}

\bibitem[\protect\citeauthoryear{Alon, Yuster, and Zwick}{Alon
  et~al\mbox{.}}{1994}]%
        {AlYuZw94}
\bibfield{author}{\bibinfo{person}{Noga Alon}, \bibinfo{person}{Raphy Yuster},
  {and} \bibinfo{person}{Uri Zwick}.} \bibinfo{year}{1994}\natexlab{}.
\newblock \showarticletitle{Color-coding: A New Method for Finding Simple
  Paths, Cycles and Other Small Subgraphs Within Large Graphs}. In
  \bibinfo{booktitle}{\emph{Symposium on the Theory of Computing (STOC)}}
  (Montreal, Quebec, Canada). \bibinfo{pages}{326--335}.
\newblock
\showISBNx{0-89791-663-8}
\urldef\tempurl%
\url{https://doi.org/10.1145/195058.195179}
\showDOI{\tempurl}


\bibitem[\protect\citeauthoryear{Alvarez-Hamelin, Dall'Asta, Barrat, and
  Vespignani}{Alvarez-Hamelin et~al\mbox{.}}{2006}]%
        {alvarez06}
\bibfield{author}{\bibinfo{person}{J~Ignacio Alvarez-Hamelin},
  \bibinfo{person}{Luca Dall'Asta}, \bibinfo{person}{Alain Barrat}, {and}
  \bibinfo{person}{Alessandro Vespignani}.} \bibinfo{year}{2006}\natexlab{}.
\newblock \showarticletitle{Large scale networks fingerprinting and
  visualization using the k-core decomposition}. In
  \bibinfo{booktitle}{\emph{Advances in neural information processing
  systems}}. \bibinfo{pages}{41--50}.
\newblock


\bibitem[\protect\citeauthoryear{Andersen and Chellapilla}{Andersen and
  Chellapilla}{2009}]%
        {AnCh09}
\bibfield{author}{\bibinfo{person}{R. Andersen} {and} \bibinfo{person}{K.
  Chellapilla}.} \bibinfo{year}{2009}\natexlab{}.
\newblock \showarticletitle{Finding Dense Subgraphs with Size Bounds}. In
  \bibinfo{booktitle}{\emph{Workshop on Algorithms and Models for the Web-Graph
  (WAW)}}. \bibinfo{pages}{25--37}.
\newblock


\bibitem[\protect\citeauthoryear{Becchetti, Boldi, Castillo, and
  Gionis}{Becchetti et~al\mbox{.}}{2008}]%
        {BeBoCaGi08}
\bibfield{author}{\bibinfo{person}{L. Becchetti}, \bibinfo{person}{P. Boldi},
  \bibinfo{person}{C. Castillo}, {and} \bibinfo{person}{A. Gionis}.}
  \bibinfo{year}{2008}\natexlab{}.
\newblock \showarticletitle{Efficient semi-streaming algorithms for local
  triangle counting in massive graphs}. In \bibinfo{booktitle}{\emph{KDD'08}}.
  \bibinfo{pages}{16--24}.
\newblock
\urldef\tempurl%
\url{https://doi.org/10.1145/1401890.1401898}
\showDOI{\tempurl}


\bibitem[\protect\citeauthoryear{Bhuiyan, Rahman, Rahman, and Al~Hasan}{Bhuiyan
  et~al\mbox{.}}{2012}]%
        {BRRA12}
\bibfield{author}{\bibinfo{person}{Mansurul~A Bhuiyan},
  \bibinfo{person}{Mahmudur Rahman}, \bibinfo{person}{Mahmuda Rahman}, {and}
  \bibinfo{person}{Mohammad Al~Hasan}.} \bibinfo{year}{2012}\natexlab{}.
\newblock \showarticletitle{Guise: Uniform sampling of graphlets for large
  graph analysis}. In \bibinfo{booktitle}{\emph{2012 IEEE 12th International
  Conference on Data Mining}}. IEEE, \bibinfo{pages}{91--100}.
\newblock


\bibitem[\protect\citeauthoryear{Bordino, Donata, Gionis, and Leonardi}{Bordino
  et~al\mbox{.}}{2008}]%
        {BoDo+08}
\bibfield{author}{\bibinfo{person}{I. Bordino}, \bibinfo{person}{D. Donata},
  \bibinfo{person}{A. Gionis}, {and} \bibinfo{person}{S. Leonardi}.}
  \bibinfo{year}{2008}\natexlab{}.
\newblock \showarticletitle{Mining Large Networks with Subgraph Counting}. In
  \bibinfo{booktitle}{\emph{Proceedings of International Conference on Data
  Mining}}. \bibinfo{pages}{737--742}.
\newblock


\bibitem[\protect\citeauthoryear{Bressan, Chierichetti, Kumar, Leucci, and
  Panconesi}{Bressan et~al\mbox{.}}{2018}]%
        {BCKLP18}
\bibfield{author}{\bibinfo{person}{Marco Bressan}, \bibinfo{person}{Flavio
  Chierichetti}, \bibinfo{person}{Ravi Kumar}, \bibinfo{person}{Stefano
  Leucci}, {and} \bibinfo{person}{Alessandro Panconesi}.}
  \bibinfo{year}{2018}\natexlab{}.
\newblock \showarticletitle{Motif Counting Beyond Five Nodes}.
\newblock \bibinfo{journal}{\emph{ACM Transactions on Knowledge Discovery from
  Data (TKDD)}} \bibinfo{volume}{12}, \bibinfo{number}{4}
  (\bibinfo{year}{2018}), \bibinfo{pages}{48}.
\newblock


\bibitem[\protect\citeauthoryear{Chen and Saad}{Chen and Saad}{2010}]%
        {chen10}
\bibfield{author}{\bibinfo{person}{Jie Chen} {and} \bibinfo{person}{Yousef
  Saad}.} \bibinfo{year}{2010}\natexlab{}.
\newblock \showarticletitle{Dense subgraph extraction with application to
  community detection}.
\newblock \bibinfo{journal}{\emph{IEEE Transactions on knowledge and data
  engineering}} \bibinfo{volume}{24}, \bibinfo{number}{7}
  (\bibinfo{year}{2010}), \bibinfo{pages}{1216--1230}.
\newblock


\bibitem[\protect\citeauthoryear{Chiba and Nishizeki}{Chiba and
  Nishizeki}{1985}]%
        {ChNi85}
\bibfield{author}{\bibinfo{person}{Norishige Chiba} {and}
  \bibinfo{person}{Takao Nishizeki}.} \bibinfo{year}{1985}\natexlab{}.
\newblock \showarticletitle{Arboricity and subgraph listing algorithms}.
\newblock \bibinfo{journal}{\emph{SIAM J. Comput.}}  \bibinfo{volume}{14}
  (\bibinfo{year}{1985}), \bibinfo{pages}{210--223}.
\newblock
Issue 1.
\showISSN{0097-5397}
\urldef\tempurl%
\url{https://doi.org/10.1137/0214017}
\showDOI{\tempurl}


\bibitem[\protect\citeauthoryear{Curticapean, Dell, and Marx}{Curticapean
  et~al\mbox{.}}{2017}]%
        {CDM17}
\bibfield{author}{\bibinfo{person}{Radu Curticapean}, \bibinfo{person}{Holger
  Dell}, {and} \bibinfo{person}{D{\'a}niel Marx}.}
  \bibinfo{year}{2017}\natexlab{}.
\newblock \showarticletitle{Homomorphisms are a good basis for counting small
  subgraphs}. In \bibinfo{booktitle}{\emph{Proceedings of the 49th Annual ACM
  SIGACT Symposium on Theory of Computing}}. ACM, \bibinfo{pages}{210--223}.
\newblock


\bibitem[\protect\citeauthoryear{Danisch, Balalau, and Sozio}{Danisch
  et~al\mbox{.}}{2018}]%
        {DBS18}
\bibfield{author}{\bibinfo{person}{Maximilien Danisch}, \bibinfo{person}{Oana
  Balalau}, {and} \bibinfo{person}{Mauro Sozio}.}
  \bibinfo{year}{2018}\natexlab{}.
\newblock \showarticletitle{Listing k-cliques in Sparse Real-World Graphs}. In
  \bibinfo{booktitle}{\emph{Proceedings of the 2018 World Wide Web Conference
  on World Wide Web}}. International World Wide Web Conferences Steering
  Committee, \bibinfo{pages}{589--598}.
\newblock


\bibitem[\protect\citeauthoryear{Dubhashi and Panconesi}{Dubhashi and
  Panconesi}{2009}]%
        {DuPa09}
\bibfield{author}{\bibinfo{person}{Devdatt Dubhashi} {and}
  \bibinfo{person}{Alessandro Panconesi}.} \bibinfo{year}{2009}\natexlab{}.
\newblock \bibinfo{booktitle}{\emph{Concentration of Measure for the Analysis
  of Randomized Algorithms}}.
\newblock \bibinfo{publisher}{Cambridge University Press}.
\newblock


\bibitem[\protect\citeauthoryear{Elenberg, Shanmugam, Borokhovich, and
  Dimakis}{Elenberg et~al\mbox{.}}{2016}]%
        {EEKBD16}
\bibfield{author}{\bibinfo{person}{Ethan~R Elenberg},
  \bibinfo{person}{Karthikeyan Shanmugam}, \bibinfo{person}{Michael
  Borokhovich}, {and} \bibinfo{person}{Alexandros~G Dimakis}.}
  \bibinfo{year}{2016}\natexlab{}.
\newblock \showarticletitle{Distributed estimation of graph 4-profiles}. In
  \bibinfo{booktitle}{\emph{Proceedings of the 25th International Conference on
  World Wide Web}}. International World Wide Web Conferences Steering
  Committee, \bibinfo{pages}{483--493}.
\newblock


\bibitem[\protect\citeauthoryear{Finocchi, Finocchi, and Fusco}{Finocchi
  et~al\mbox{.}}{2015}]%
        {FFF15}
\bibfield{author}{\bibinfo{person}{Irene Finocchi}, \bibinfo{person}{Marco
  Finocchi}, {and} \bibinfo{person}{Emanuele~G. Fusco}.}
  \bibinfo{year}{2015}\natexlab{}.
\newblock \showarticletitle{Clique Counting in MapReduce: Algorithms and
  Experiments}.
\newblock \bibinfo{journal}{\emph{{ACM} Journal of Experimental Algorithmics}}
  \bibinfo{volume}{20} (\bibinfo{year}{2015}).
\newblock
\urldef\tempurl%
\url{https://doi.org/10.1145/2794080}
\showDOI{\tempurl}


\bibitem[\protect\citeauthoryear{Fratkin, Naughton, Brutlag, and
  Batzoglou}{Fratkin et~al\mbox{.}}{2006}]%
        {fratkin06}
\bibfield{author}{\bibinfo{person}{Eugene Fratkin}, \bibinfo{person}{Brian~T
  Naughton}, \bibinfo{person}{Douglas~L Brutlag}, {and}
  \bibinfo{person}{Serafim Batzoglou}.} \bibinfo{year}{2006}\natexlab{}.
\newblock \showarticletitle{MotifCut: regulatory motifs finding with maximum
  density subgraphs}.
\newblock \bibinfo{journal}{\emph{Bioinformatics}} \bibinfo{volume}{22},
  \bibinfo{number}{14} (\bibinfo{year}{2006}), \bibinfo{pages}{e150--e157}.
\newblock


\bibitem[\protect\citeauthoryear{Han and Sethu}{Han and Sethu}{2016}]%
        {HaSe16}
\bibfield{author}{\bibinfo{person}{Guyue Han} {and} \bibinfo{person}{Harish
  Sethu}.} \bibinfo{year}{2016}\natexlab{}.
\newblock \showarticletitle{Waddling random walk: Fast and accurate mining of
  motif statistics in large graphs}. In \bibinfo{booktitle}{\emph{Data Mining
  (ICDM), 2016 IEEE 16th International Conference on}}. IEEE,
  \bibinfo{pages}{181--190}.
\newblock


\bibitem[\protect\citeauthoryear{Ho{\v{c}}evar and Dem{\v{s}}ar}{Ho{\v{c}}evar
  and Dem{\v{s}}ar}{2017}]%
        {HoDe17}
\bibfield{author}{\bibinfo{person}{Toma{\v{z}} Ho{\v{c}}evar} {and}
  \bibinfo{person}{Janez Dem{\v{s}}ar}.} \bibinfo{year}{2017}\natexlab{}.
\newblock \showarticletitle{Combinatorial algorithm for counting small induced
  graphs and orbits}.
\newblock \bibinfo{journal}{\emph{PloS one}} \bibinfo{volume}{12},
  \bibinfo{number}{2} (\bibinfo{year}{2017}), \bibinfo{pages}{e0171428}.
\newblock


\bibitem[\protect\citeauthoryear{Holland and Leinhardt}{Holland and
  Leinhardt}{1970}]%
        {HoLe70}
\bibfield{author}{\bibinfo{person}{P. Holland} {and} \bibinfo{person}{S.
  Leinhardt}.} \bibinfo{year}{1970}\natexlab{}.
\newblock \showarticletitle{A method for detecting structure in sociometric
  data}.
\newblock \bibinfo{journal}{\emph{Amer. J. Sociology}}  \bibinfo{volume}{76}
  (\bibinfo{year}{1970}), \bibinfo{pages}{492--513}.
\newblock


\bibitem[\protect\citeauthoryear{Jain and Seshadhri}{Jain and
  Seshadhri}{2017}]%
        {JaSe17}
\bibfield{author}{\bibinfo{person}{Shweta Jain} {and} \bibinfo{person}{C
  Seshadhri}.} \bibinfo{year}{2017}\natexlab{}.
\newblock \showarticletitle{A Fast and Provable Method for Estimating Clique
  Counts Using Tur{\'a}n's Theorem}. In \bibinfo{booktitle}{\emph{Proceedings
  of the 26th International Conference on World Wide Web}}. International World
  Wide Web Conferences Steering Committee, \bibinfo{pages}{441--449}.
\newblock


\bibitem[\protect\citeauthoryear{Jha, Seshadhri, and Pinar}{Jha
  et~al\mbox{.}}{2015}]%
        {JhSePi15}
\bibfield{author}{\bibinfo{person}{M. Jha}, \bibinfo{person}{C. Seshadhri},
  {and} \bibinfo{person}{A. Pinar}.} \bibinfo{year}{2015}\natexlab{}.
\newblock \showarticletitle{Path Sampling: A Fast and Provable Method for
  Estimating 4-Vertex Subgraph Counts}. In \bibinfo{booktitle}{\emph{World Wide
  Web (WWW)}}. \bibinfo{pages}{495--505}.
\newblock


\bibitem[\protect\citeauthoryear{Kane, Mehlhorn, Sauerwald, and Sun}{Kane
  et~al\mbox{.}}{2012}]%
        {KMSS12}
\bibfield{author}{\bibinfo{person}{Daniel~M Kane}, \bibinfo{person}{Kurt
  Mehlhorn}, \bibinfo{person}{Thomas Sauerwald}, {and} \bibinfo{person}{He
  Sun}.} \bibinfo{year}{2012}\natexlab{}.
\newblock \showarticletitle{Counting arbitrary subgraphs in data streams}. In
  \bibinfo{booktitle}{\emph{International Colloquium on Automata, Languages,
  and Programming}}. Springer, \bibinfo{pages}{598--609}.
\newblock


\bibitem[\protect\citeauthoryear{Kumar, Raghavan, Rajagopalan, and
  Tomkins}{Kumar et~al\mbox{.}}{1999}]%
        {kumar99}
\bibfield{author}{\bibinfo{person}{Ravi Kumar}, \bibinfo{person}{Prabhakar
  Raghavan}, \bibinfo{person}{Sridhar Rajagopalan}, {and}
  \bibinfo{person}{Andrew Tomkins}.} \bibinfo{year}{1999}\natexlab{}.
\newblock \showarticletitle{Trawling the Web for emerging cyber-communities}.
\newblock \bibinfo{journal}{\emph{Computer networks}} \bibinfo{volume}{31},
  \bibinfo{number}{11-16} (\bibinfo{year}{1999}), \bibinfo{pages}{1481--1493}.
\newblock


\bibitem[\protect\citeauthoryear{Liu and Wong}{Liu and Wong}{2008}]%
        {LW08}
\bibfield{author}{\bibinfo{person}{Guimei Liu} {and} \bibinfo{person}{Limsoon
  Wong}.} \bibinfo{year}{2008}\natexlab{}.
\newblock \showarticletitle{Effective pruning techniques for mining
  quasi-cliques}. In \bibinfo{booktitle}{\emph{Joint European conference on
  machine learning and knowledge discovery in databases}}. Springer,
  \bibinfo{pages}{33--49}.
\newblock


\bibitem[\protect\citeauthoryear{Matula and Beck}{Matula and Beck}{1983}]%
        {MB83}
\bibfield{author}{\bibinfo{person}{David~W Matula} {and}
  \bibinfo{person}{Leland~L Beck}.} \bibinfo{year}{1983}\natexlab{}.
\newblock \showarticletitle{Smallest-last ordering and clustering and graph
  coloring algorithms}.
\newblock \bibinfo{journal}{\emph{Journal of the ACM (JACM)}}
  \bibinfo{volume}{30}, \bibinfo{number}{3} (\bibinfo{year}{1983}),
  \bibinfo{pages}{417--427}.
\newblock


\bibitem[\protect\citeauthoryear{Milo, {Shen-Orr}, Itzkovitz, Kashtan,
  Chklovskii, and Alon}{Milo et~al\mbox{.}}{2002}]%
        {Milo2002}
\bibfield{author}{\bibinfo{person}{R. Milo}, \bibinfo{person}{S. {Shen-Orr}},
  \bibinfo{person}{S. Itzkovitz}, \bibinfo{person}{N. Kashtan},
  \bibinfo{person}{D. Chklovskii}, {and} \bibinfo{person}{U. Alon}.}
  \bibinfo{year}{2002}\natexlab{}.
\newblock \showarticletitle{Network motifs: Simple building blocks of complex
  networks}.
\newblock \bibinfo{journal}{\emph{Science}} \bibinfo{volume}{298},
  \bibinfo{number}{5594} (\bibinfo{year}{2002}), \bibinfo{pages}{824--827}.
\newblock


\bibitem[\protect\citeauthoryear{Paranjape, Benson, and Leskovec}{Paranjape
  et~al\mbox{.}}{2017}]%
        {PBL17}
\bibfield{author}{\bibinfo{person}{Ashwin Paranjape}, \bibinfo{person}{Austin~R
  Benson}, {and} \bibinfo{person}{Jure Leskovec}.}
  \bibinfo{year}{2017}\natexlab{}.
\newblock \showarticletitle{Motifs in temporal networks}. In
  \bibinfo{booktitle}{\emph{Proceedings of the Tenth ACM International
  Conference on Web Search and Data Mining}}. ACM, \bibinfo{pages}{601--610}.
\newblock


\bibitem[\protect\citeauthoryear{Pattillo, Youssef, and Butenko}{Pattillo
  et~al\mbox{.}}{2012}]%
        {PYB12}
\bibfield{author}{\bibinfo{person}{Jeffrey Pattillo}, \bibinfo{person}{Nataly
  Youssef}, {and} \bibinfo{person}{Sergiy Butenko}.}
  \bibinfo{year}{2012}\natexlab{}.
\newblock \showarticletitle{Clique relaxation models in social network
  analysis}.
\newblock In \bibinfo{booktitle}{\emph{Handbook of Optimization in Complex
  Networks}}. \bibinfo{publisher}{Springer}, \bibinfo{pages}{143--162}.
\newblock


\bibitem[\protect\citeauthoryear{Pinar, Seshadhri, and Vishal}{Pinar
  et~al\mbox{.}}{2017}]%
        {PiSeVi17}
\bibfield{author}{\bibinfo{person}{Ali Pinar}, \bibinfo{person}{C Seshadhri},
  {and} \bibinfo{person}{Vaidyanathan Vishal}.}
  \bibinfo{year}{2017}\natexlab{}.
\newblock \showarticletitle{Escape: Efficiently counting all 5-vertex
  subgraphs}. In \bibinfo{booktitle}{\emph{Proceedings of the 26th
  International Conference on World Wide Web}}. International World Wide Web
  Conferences Steering Committee, \bibinfo{pages}{1431--1440}.
\newblock


\bibitem[\protect\citeauthoryear{Pr{\v{z}}ulj}{Pr{\v{z}}ulj}{2007}]%
        {Pr07}
\bibfield{author}{\bibinfo{person}{Nata{\v{s}}a Pr{\v{z}}ulj}.}
  \bibinfo{year}{2007}\natexlab{}.
\newblock \showarticletitle{Biological network comparison using graphlet degree
  distribution}.
\newblock \bibinfo{journal}{\emph{Bioinformatics}} \bibinfo{volume}{23},
  \bibinfo{number}{2} (\bibinfo{year}{2007}), \bibinfo{pages}{e177--e183}.
\newblock


\bibitem[\protect\citeauthoryear{Rossi, Gleich, and Gebremedhin}{Rossi
  et~al\mbox{.}}{2015}]%
        {RossiG15}
\bibfield{author}{\bibinfo{person}{Ryan~A Rossi}, \bibinfo{person}{David~F
  Gleich}, {and} \bibinfo{person}{Assefaw~H Gebremedhin}.}
  \bibinfo{year}{2015}\natexlab{}.
\newblock \showarticletitle{Parallel Maximum Clique Algorithms with
  Applications to Network Analysis}.
\newblock \bibinfo{journal}{\emph{SIAM Journal on Scientific Computing}}
  \bibinfo{volume}{37}, \bibinfo{number}{5} (\bibinfo{year}{2015}),
  \bibinfo{pages}{C589--C616}.
\newblock


\bibitem[\protect\citeauthoryear{Sariyuce, Seshadhri, Pinar, and
  Catalyurek}{Sariyuce et~al\mbox{.}}{2015}]%
        {sariyuce15}
\bibfield{author}{\bibinfo{person}{Ahmet~Erdem Sariyuce}, \bibinfo{person}{C
  Seshadhri}, \bibinfo{person}{Ali Pinar}, {and} \bibinfo{person}{Umit~V
  Catalyurek}.} \bibinfo{year}{2015}\natexlab{}.
\newblock \showarticletitle{Finding the hierarchy of dense subgraphs using
  nucleus decompositions}. In \bibinfo{booktitle}{\emph{Proceedings of the 24th
  International Conference on World Wide Web}}. International World Wide Web
  Conferences Steering Committee, \bibinfo{pages}{927--937}.
\newblock


\bibitem[\protect\citeauthoryear{Sariy{\"{u}}ce, Seshadhri, Pinar, and
  {\c{C}}ataly{\"{u}}rek}{Sariy{\"{u}}ce et~al\mbox{.}}{2015}]%
        {SaSePi14}
\bibfield{author}{\bibinfo{person}{Ahmet~Erdem Sariy{\"{u}}ce},
  \bibinfo{person}{C. Seshadhri}, \bibinfo{person}{Ali Pinar}, {and}
  \bibinfo{person}{{\"{U}}mit~V. {\c{C}}ataly{\"{u}}rek}.}
  \bibinfo{year}{2015}\natexlab{}.
\newblock \showarticletitle{Finding the Hierarchy of Dense Subgraphs using
  Nucleus Decompositions}.
\newblock  (\bibinfo{year}{2015}), \bibinfo{pages}{927--937}.
\newblock


\bibitem[\protect\citeauthoryear{Seshadhri, Kolda, and Pinar}{Seshadhri
  et~al\mbox{.}}{2012}]%
        {SeKoPi11}
\bibfield{author}{\bibinfo{person}{C. Seshadhri}, \bibinfo{person}{Tamara~G.
  Kolda}, {and} \bibinfo{person}{Ali Pinar}.} \bibinfo{year}{2012}\natexlab{}.
\newblock \showarticletitle{Community structure and scale-free collections of
  {Erd\"os-R\'enyi} graphs}.
\newblock \bibinfo{journal}{\emph{Physical Review E}} \bibinfo{volume}{85},
  \bibinfo{number}{5} (\bibinfo{date}{May} \bibinfo{year}{2012}),
  \bibinfo{pages}{056109}.
\newblock
\urldef\tempurl%
\url{https://doi.org/10.1103/PhysRevE.85.056109}
\showDOI{\tempurl}


\bibitem[\protect\citeauthoryear{Silva, Paredes, and Ribeiro}{Silva
  et~al\mbox{.}}{2017}]%
        {SPR17}
\bibfield{author}{\bibinfo{person}{Miguel~EP Silva}, \bibinfo{person}{Pedro
  Paredes}, {and} \bibinfo{person}{Pedro Ribeiro}.}
  \bibinfo{year}{2017}\natexlab{}.
\newblock \showarticletitle{Network motifs detection using random networks with
  prescribed subgraph frequencies}. In \bibinfo{booktitle}{\emph{Workshop on
  Complex Networks CompleNet}}. Springer, \bibinfo{pages}{17--29}.
\newblock


\bibitem[\protect\citeauthoryear{Sizemore, Giusti, and Bassett}{Sizemore
  et~al\mbox{.}}{2016}]%
        {SGB16}
\bibfield{author}{\bibinfo{person}{Ann Sizemore}, \bibinfo{person}{Chad
  Giusti}, {and} \bibinfo{person}{Danielle~S. Bassett}.}
  \bibinfo{year}{2016}\natexlab{}.
\newblock \showarticletitle{Classification of weighted networks through
  mesoscale homological features}.
\newblock \bibinfo{journal}{\emph{Journal of Complex Networks}}
  \bibinfo{volume}{10.1093} (\bibinfo{year}{2016}).
\newblock


\bibitem[\protect\citeauthoryear{SNAP}{SNAP}{[n.d.]}]%
        {Snap}
SNAP \bibinfo{year}{[n.d.]}\natexlab{}.
\newblock \bibinfo{title}{{Stanford Network Analysis Project (SNAP)}}.
\newblock
\newblock
\newblock
\shownote{Available at \url{http://snap.stanford.edu/}.}


\bibitem[\protect\citeauthoryear{Tang, Zhang, Yao, Li, Zhang, and Su}{Tang
  et~al\mbox{.}}{2008}]%
        {Tang08}
\bibfield{author}{\bibinfo{person}{Jie Tang}, \bibinfo{person}{Jing Zhang},
  \bibinfo{person}{Limin Yao}, \bibinfo{person}{Juanzi Li}, \bibinfo{person}{Li
  Zhang}, {and} \bibinfo{person}{Zhong Su}.} \bibinfo{year}{2008}\natexlab{}.
\newblock \showarticletitle{ArnetMiner: Extraction and Mining of Academic
  Social Networks}. In \bibinfo{booktitle}{\emph{KDD'08}}.
  \bibinfo{pages}{990--998}.
\newblock


\bibitem[\protect\citeauthoryear{Tsourakakis, Bonchi, Gionis, Gullo, and
  Tsiarli}{Tsourakakis et~al\mbox{.}}{2013}]%
        {Tsourakakis13}
\bibfield{author}{\bibinfo{person}{C. Tsourakakis}, \bibinfo{person}{F.
  Bonchi}, \bibinfo{person}{A. Gionis}, \bibinfo{person}{F. Gullo}, {and}
  \bibinfo{person}{M. Tsiarli}.} \bibinfo{year}{2013}\natexlab{}.
\newblock \showarticletitle{Denser Than the Densest Subgraph: Extracting
  Optimal Quasi-cliques with Quality Guarantees}. In
  \bibinfo{booktitle}{\emph{Knowledge Data and Discovery (KDD)}}.
\newblock


\bibitem[\protect\citeauthoryear{Tsourakakis}{Tsourakakis}{2015}]%
        {Ts15}
\bibfield{author}{\bibinfo{person}{Charalampos~E. Tsourakakis}.}
  \bibinfo{year}{2015}\natexlab{}.
\newblock \showarticletitle{The K-clique Densest Subgraph Problem}. In
  \bibinfo{booktitle}{\emph{Proceedings of the Conference on World Wide Web
  {WWW}}}. \bibinfo{pages}{1122--1132}.
\newblock
\urldef\tempurl%
\url{https://doi.org/10.1145/2736277.2741098}
\showDOI{\tempurl}


\bibitem[\protect\citeauthoryear{Tsourakakis, Pachocki, and
  Mitzenmacher}{Tsourakakis et~al\mbox{.}}{2016}]%
        {TsPaMi16}
\bibfield{author}{\bibinfo{person}{Charalampos~E. Tsourakakis},
  \bibinfo{person}{Jakub~W. Pachocki}, {and} \bibinfo{person}{Michael
  Mitzenmacher}.} \bibinfo{year}{2016}\natexlab{}.
\newblock \showarticletitle{Scalable motif-aware graph clustering}.
\newblock \bibinfo{journal}{\emph{CoRR}}  \bibinfo{volume}{abs/1606.06235}
  (\bibinfo{year}{2016}).
\newblock
\urldef\tempurl%
\url{http://arxiv.org/abs/1606.06235}
\showURL{%
\tempurl}


\bibitem[\protect\citeauthoryear{Ugander, Backstrom, and Kleinberg}{Ugander
  et~al\mbox{.}}{2013}]%
        {UganderBK13}
\bibfield{author}{\bibinfo{person}{Johan Ugander}, \bibinfo{person}{Lars
  Backstrom}, {and} \bibinfo{person}{Jon~M. Kleinberg}.}
  \bibinfo{year}{2013}\natexlab{}.
\newblock \showarticletitle{Subgraph frequencies: mapping the empirical and
  extremal geography of large graph collections}. In
  \bibinfo{booktitle}{\emph{WWW}}, \bibfield{editor}{\bibinfo{person}{Daniel
  Schwabe}, \bibinfo{person}{Virg\'{\i}lio A.~F. Almeida},
  \bibinfo{person}{Hartmut Glaser}, \bibinfo{person}{Ricardo~A. Baeza-Yates},
  {and} \bibinfo{person}{Sue~B. Moon}} (Eds.).
  \bibinfo{publisher}{International World Wide Web Conferences Steering
  Committee / ACM}, \bibinfo{pages}{1307--1318}.
\newblock
\showISBNx{978-1-4503-2035-1}


\bibitem[\protect\citeauthoryear{Wang, Lui, Ribeiro, Towsley, Zhao, and
  Guan}{Wang et~al\mbox{.}}{2014}]%
        {WLRTZG14}
\bibfield{author}{\bibinfo{person}{Pinghui Wang}, \bibinfo{person}{John Lui},
  \bibinfo{person}{Bruno Ribeiro}, \bibinfo{person}{Don Towsley},
  \bibinfo{person}{Junzhou Zhao}, {and} \bibinfo{person}{Xiaohong Guan}.}
  \bibinfo{year}{2014}\natexlab{}.
\newblock \showarticletitle{Efficiently estimating motif statistics of large
  networks}.
\newblock \bibinfo{journal}{\emph{ACM Transactions on Knowledge Discovery from
  Data (TKDD)}} \bibinfo{volume}{9}, \bibinfo{number}{2}
  (\bibinfo{year}{2014}), \bibinfo{pages}{8}.
\newblock


\bibitem[\protect\citeauthoryear{Wang, Zhao, Zhang, Li, Cheng, Lui, Towsley,
  Tao, and Guan}{Wang et~al\mbox{.}}{2018}]%
        {WZ+18}
\bibfield{author}{\bibinfo{person}{Pinghui Wang}, \bibinfo{person}{Junzhou
  Zhao}, \bibinfo{person}{Xiangliang Zhang}, \bibinfo{person}{Zhenguo Li},
  \bibinfo{person}{Jiefeng Cheng}, \bibinfo{person}{John~CS Lui},
  \bibinfo{person}{Don Towsley}, \bibinfo{person}{Jing Tao}, {and}
  \bibinfo{person}{Xiaohong Guan}.} \bibinfo{year}{2018}\natexlab{}.
\newblock \showarticletitle{MOSS-5: A fast method of approximating counts of
  5-node graphlets in large graphs}.
\newblock \bibinfo{journal}{\emph{IEEE Transactions on Knowledge and Data
  Engineering}} \bibinfo{volume}{30}, \bibinfo{number}{1}
  (\bibinfo{year}{2018}), \bibinfo{pages}{73--86}.
\newblock


\bibitem[\protect\citeauthoryear{Wernicke}{Wernicke}{2006}]%
        {Wernicke06}
\bibfield{author}{\bibinfo{person}{Sebastian Wernicke}.}
  \bibinfo{year}{2006}\natexlab{}.
\newblock \showarticletitle{Efficient Detection of Network Motifs}.
\newblock \bibinfo{journal}{\emph{IEEE/ACM Trans. Comput. Biology Bioinform.}}
  \bibinfo{volume}{3}, \bibinfo{number}{4} (\bibinfo{year}{2006}),
  \bibinfo{pages}{347--359}.
\newblock


\bibitem[\protect\citeauthoryear{Yin, Benson, Leskovec, and Gleich}{Yin
  et~al\mbox{.}}{2017}]%
        {YBLG17}
\bibfield{author}{\bibinfo{person}{Hao Yin}, \bibinfo{person}{Austin~R Benson},
  \bibinfo{person}{Jure Leskovec}, {and} \bibinfo{person}{David~F Gleich}.}
  \bibinfo{year}{2017}\natexlab{}.
\newblock \showarticletitle{Local higher-order graph clustering}. In
  \bibinfo{booktitle}{\emph{Proceedings of the 23rd ACM SIGKDD International
  Conference on Knowledge Discovery and Data Mining}}. ACM,
  \bibinfo{pages}{555--564}.
\newblock


\bibitem[\protect\citeauthoryear{Yu, Paccanaro, Trifonov, and Gerstein}{Yu
  et~al\mbox{.}}{2006}]%
        {YPTG06}
\bibfield{author}{\bibinfo{person}{Haiyuan Yu}, \bibinfo{person}{Alberto
  Paccanaro}, \bibinfo{person}{Valery Trifonov}, {and} \bibinfo{person}{Mark
  Gerstein}.} \bibinfo{year}{2006}\natexlab{}.
\newblock \showarticletitle{Predicting interactions in protein networks by
  completing defective cliques}.
\newblock \bibinfo{journal}{\emph{Bioinformatics}} \bibinfo{volume}{22},
  \bibinfo{number}{7} (\bibinfo{year}{2006}), \bibinfo{pages}{823--829}.
\newblock


\end{thebibliography}

\end{document}